\documentclass[journal,draftclsnofoot,onecolumn,12pt]{IEEEtran}
\usepackage{array}
\newcolumntype{P}[1]{>{\centering\arraybackslash}p{#1}}
\newcolumntype{M}[1]{>{\centering\arraybackslash}m{#1}}
\usepackage{amsthm,amssymb,amsmath,graphicx,multirow,color,amsfonts}%
\usepackage[update,prepend]{epstopdf}
\usepackage[latin1]{inputenc}
\usepackage{tikz}
\usepackage{bbm} 
\usepackage{pdfpages}
\usepackage{subfig}
\usepackage{comment}
\usepackage{multicol}

\usepackage{cite}
\usepackage[justification=centering]{caption}
\usepackage{textcomp}
\usepackage{psfrag}
\usepackage{arydshln}
\usepackage{url}
\usepackage{soul}
\usepackage{graphicx,color}
\usepackage[nolist]{acronym}
\usepackage{algorithm,algorithmic} 

\usepackage{mathtools,lipsum}
\usepackage{cuted}
\usepackage{amsmath}
\newtheorem{proposition}{Proposition} 
\usepackage{mathrsfs}

\usepackage[capitalise]{cleveref}
\Crefname{equation}{Eq.\!}{Eqs.\!}
\Crefname{figure}{Fig.\!}{Figs.\!}
\Crefname{tabular}{Tab.\!}{Tabs.\!}
\Crefname{section}{Section\!}{Sections.\!}

\def\nb0{{\mathbf{0}}}
\def\nb1{{\mathbf{1}}}

\newtheorem{lemma}{Lemma}

\newtheorem{definition}{Definition}

\newtheorem{theorem}{Theorem}

\newenvironment{sequation}{
\begin{equation}\small}{\end{equation}
}

\begin{document}
\graphicspath{{./Figures/}}
	\begin{acronym}

\acro{5G-NR}{5G New Radio}
\acro{3GPP}{3rd Generation Partnership Project}
\acro{ABS}{aerial base station}
\acro{AC}{address coding}
\acro{ACF}{autocorrelation function}
\acro{ACR}{autocorrelation receiver}
\acro{ADC}{analog-to-digital converter}
\acrodef{aic}[AIC]{Analog-to-Information Converter}     
\acro{AIC}[AIC]{Akaike information criterion}
\acro{aric}[ARIC]{asymmetric restricted isometry constant}
\acro{arip}[ARIP]{asymmetric restricted isometry property}

\acro{ARQ}{Automatic Repeat Request}
\acro{AUB}{asymptotic union bound}
\acrodef{awgn}[AWGN]{Additive White Gaussian Noise}     
\acro{AWGN}{additive white Gaussian noise}

\acro{APSK}[PSK]{asymmetric PSK} 

\acro{waric}[AWRICs]{asymmetric weak restricted isometry constants}
\acro{warip}[AWRIP]{asymmetric weak restricted isometry property}
\acro{BCH}{Bose, Chaudhuri, and Hocquenghem}        
\acro{BCHC}[BCHSC]{BCH based source coding}
\acro{BEP}{bit error probability}
\acro{BFC}{block fading channel}
\acro{BG}[BG]{Bernoulli-Gaussian}
\acro{BGG}{Bernoulli-Generalized Gaussian}
\acro{BPAM}{binary pulse amplitude modulation}
\acro{BPDN}{Basis Pursuit Denoising}
\acro{BPPM}{binary pulse position modulation}
\acro{BPSK}{Binary Phase Shift Keying}
\acro{BPZF}{bandpass zonal filter}
\acro{BSC}{binary symmetric channels}              
\acro{BU}[BU]{Bernoulli-uniform}
\acro{BER}{bit error rate}
\acro{BS}{base station}
\acro{BW}{BandWidth}
\acro{BLLL}{ binary log-linear learning }

\acro{CP}{Cyclic Prefix}
\acrodef{cdf}[CDF]{cumulative distribution function}   
\acro{CDF}{Cumulative Distribution Function}
\acrodef{c.d.f.}[CDF]{cumulative distribution function}
\acro{CCDF}{complementary cumulative distribution function}
\acrodef{ccdf}[CCDF]{complementary CDF}               
\acrodef{c.c.d.f.}[CCDF]{complementary cumulative distribution function}
\acro{CD}{cooperative diversity}

\acro{CDMA}{Code Division Multiple Access}
\acro{ch.f.}{characteristic function}
\acro{CIR}{channel impulse response}
\acro{cosamp}[CoSaMP]{compressive sampling matching pursuit}
\acro{CR}{cognitive radio}
\acro{cs}[CS]{compressed sensing}                   
\acrodef{cscapital}[CS]{Compressed sensing} 
\acrodef{CS}[CS]{compressed sensing}
\acro{CSI}{channel state information}
\acro{CCSDS}{consultative committee for space data systems}
\acro{CC}{convolutional coding}
\acro{Covid19}[COVID-19]{Coronavirus disease}

\acro{DAA}{detect and avoid}
\acro{DAB}{digital audio broadcasting}
\acro{DCT}{discrete cosine transform}
\acro{dft}[DFT]{discrete Fourier transform}
\acro{DR}{distortion-rate}
\acro{DS}{direct sequence}
\acro{DS-SS}{direct-sequence spread-spectrum}
\acro{DTR}{differential transmitted-reference}
\acro{DVB-H}{digital video broadcasting\,--\,handheld}
\acro{DVB-T}{digital video broadcasting\,--\,terrestrial}
\acro{DL}{DownLink}
\acro{DSSS}{Direct Sequence Spread Spectrum}
\acro{DFT-s-OFDM}{Discrete Fourier Transform-spread-Orthogonal Frequency Division Multiplexing}
\acro{DAS}{Distributed Antenna System}
\acro{DNA}{DeoxyriboNucleic Acid}

\acro{EC}{European Commission}
\acro{EED}[EED]{exact eigenvalues distribution}
\acro{EIRP}{Equivalent Isotropically Radiated Power}
\acro{ELP}{equivalent low-pass}
\acro{eMBB}{Enhanced Mobile Broadband}
\acro{EMF}{ElectroMagnetic Field}
\acro{EU}{European union}
\acro{EI}{Exposure Index}
\acro{eICIC}{enhanced Inter-Cell Interference Coordination}

\acro{FC}[FC]{fusion center}
\acro{FCC}{Federal Communications Commission}
\acro{FEC}{forward error correction}
\acro{FFT}{fast Fourier transform}
\acro{FH}{frequency-hopping}
\acro{FH-SS}{frequency-hopping spread-spectrum}
\acrodef{FS}{Frame synchronization}
\acro{FSsmall}[FS]{frame synchronization}  
\acro{FDMA}{Frequency Division Multiple Access}

\acro{GA}{Gaussian approximation}
\acro{GF}{Galois field }
\acro{GG}{Generalized-Gaussian}
\acro{GIC}[GIC]{generalized information criterion}
\acro{GLRT}{generalized likelihood ratio test}
\acro{GPS}{Global Positioning System}
\acro{GMSK}{Gaussian Minimum Shift Keying}
\acro{GSMA}{Global System for Mobile communications Association}
\acro{GS}{ground station}
\acro{GMG}{ Grid-connected MicroGeneration}

\acro{HAP}{high altitude platform}
\acro{HetNet}{Heterogeneous network}

\acro{IDR}{information distortion-rate}
\acro{IFFT}{inverse fast Fourier transform}
\acro{iht}[IHT]{iterative hard thresholding}
\acro{i.i.d.}{independent, identically distributed}
\acro{IoT}{Internet of Things}                      
\acro{IR}{impulse radio}
\acro{lric}[LRIC]{lower restricted isometry constant}
\acro{lrict}[LRICt]{lower restricted isometry constant threshold}
\acro{ISI}{intersymbol interference}
\acro{ITU}{International Telecommunication Union}
\acro{ICNIRP}{International Commission on Non-Ionizing Radiation Protection}
\acro{IEEE}{Institute of Electrical and Electronics Engineers}
\acro{ICES}{IEEE international committee on electromagnetic safety}
\acro{IEC}{International Electrotechnical Commission}
\acro{IARC}{International Agency on Research on Cancer}
\acro{IS-95}{Interim Standard 95}

\acro{KPI}{Key Performance Indicator}

\acro{LEO}{low earth orbit}
\acro{LF}{likelihood function}
\acro{LLF}{log-likelihood function}
\acro{LLR}{log-likelihood ratio}
\acro{LLRT}{log-likelihood ratio test}
\acro{LoS}{Line-of-Sight}
\acro{LRT}{likelihood ratio test}
\acro{wlric}[LWRIC]{lower weak restricted isometry constant}
\acro{wlrict}[LWRICt]{LWRIC threshold}
\acro{LPWAN}{Low Power Wide Area Network}
\acro{LoRaWAN}{Low power long Range Wide Area Network}
\acro{NLoS}{Non-Line-of-Sight}
\acro{LiFi}[Li-Fi]{light-fidelity}
 \acro{LED}{light emitting diode}
 \acro{LABS}{LoS transmission with each ABS}
 \acro{NLABS}{NLoS transmission with each ABS}

\acro{MB}{multiband}
\acro{MC}{macro cell}
\acro{MDS}{mixed distributed source}
\acro{MF}{matched filter}
\acro{m.g.f.}{moment generating function}
\acro{MI}{mutual information}
\acro{MIMO}{Multiple-Input Multiple-Output}
\acro{MISO}{multiple-input single-output}
\acrodef{maxs}[MJSO]{maximum joint support cardinality}                       
\acro{ML}[ML]{maximum likelihood}
\acro{MMSE}{minimum mean-square error}
\acro{MMV}{multiple measurement vectors}
\acrodef{MOS}{model order selection}
\acro{M-PSK}[${M}$-PSK]{$M$-ary phase shift keying}                       
\acro{M-APSK}[${M}$-PSK]{$M$-ary asymmetric PSK} 
\acro{MP}{ multi-period}
\acro{MINLP}{mixed integer non-linear programming}

\acro{M-QAM}[$M$-QAM]{$M$-ary quadrature amplitude modulation}
\acro{MRC}{maximal ratio combiner}                  
\acro{maxs}[MSO]{maximum sparsity order}                                      
\acro{M2M}{Machine-to-Machine}                                                
\acro{MUI}{multi-user interference}
\acro{mMTC}{massive Machine Type Communications}      
\acro{mm-Wave}{millimeter-wave}
\acro{MP}{mobile phone}
\acro{MPE}{maximum permissible exposure}
\acro{MAC}{media access control}
\acro{NB}{narrowband}
\acro{NBI}{narrowband interference}
\acro{NLA}{nonlinear sparse approximation}
\acro{NLOS}{Non-Line of Sight}
\acro{NTIA}{National Telecommunications and Information Administration}
\acro{NTP}{National Toxicology Program}
\acro{NHS}{National Health Service}

\acro{LOS}{Line of Sight}

\acro{OC}{optimum combining}                             
\acro{OC}{optimum combining}
\acro{ODE}{operational distortion-energy}
\acro{ODR}{operational distortion-rate}
\acro{OFDM}{Orthogonal Frequency-Division Multiplexing}
\acro{omp}[OMP]{orthogonal matching pursuit}
\acro{OSMP}[OSMP]{orthogonal subspace matching pursuit}
\acro{OQAM}{offset quadrature amplitude modulation}
\acro{OQPSK}{offset QPSK}
\acro{OFDMA}{Orthogonal Frequency-division Multiple Access}
\acro{OPEX}{Operating Expenditures}
\acro{OQPSK/PM}{OQPSK with phase modulation}

\acro{PAM}{pulse amplitude modulation}
\acro{PAR}{peak-to-average ratio}
\acrodef{pdf}[PDF]{probability density function}                      
\acro{PDF}{probability density function}
\acrodef{p.d.f.}[PDF]{probability distribution function}
\acro{PDP}{power dispersion profile}
\acro{PMF}{probability mass function}                             
\acrodef{p.m.f.}[PMF]{probability mass function}
\acro{PN}{pseudo-noise}
\acro{PPM}{pulse position modulation}
\acro{PRake}{Partial Rake}
\acro{PSD}{power spectral density}
\acro{PSEP}{pairwise synchronization error probability}
\acro{PSK}{phase shift keying}
\acro{PD}{power density}
\acro{8-PSK}[$8$-PSK]{$8$-phase shift keying}
\acro{PPP}{Poisson point process}
\acro{PCP}{Poisson cluster process}
 
\acro{FSK}{Frequency Shift Keying}

\acro{QAM}{Quadrature Amplitude Modulation}
\acro{QPSK}{Quadrature Phase Shift Keying}
\acro{OQPSK/PM}{OQPSK with phase modulator }

\acro{RD}[RD]{raw data}
\acro{RDL}{"random data limit"}
\acro{ric}[RIC]{restricted isometry constant}
\acro{rict}[RICt]{restricted isometry constant threshold}
\acro{rip}[RIP]{restricted isometry property}
\acro{ROC}{receiver operating characteristic}
\acro{rq}[RQ]{Raleigh quotient}
\acro{RS}[RS]{Reed-Solomon}
\acro{RSC}[RSSC]{RS based source coding}
\acro{r.v.}{random variable}                               
\acro{R.V.}{random vector}
\acro{RMS}{root mean square}
\acro{RFR}{radiofrequency radiation}
\acro{RIS}{Reconfigurable Intelligent Surface}
\acro{RNA}{RiboNucleic Acid}
\acro{RRM}{Radio Resource Management}
\acro{RUE}{reference user equipments}
\acro{RAT}{radio access technology}
\acro{RB}{resource block}

\acro{SA}[SA-Music]{subspace-augmented MUSIC with OSMP}
\acro{SC}{small cell}
\acro{SCBSES}[SCBSES]{Source Compression Based Syndrome Encoding Scheme}
\acro{SCM}{sample covariance matrix}
\acro{SEP}{symbol error probability}
\acro{SG}[SG]{sparse-land Gaussian model}
\acro{SIMO}{single-input multiple-output}
\acro{SINR}{signal-to-interference plus noise ratio}
\acro{SIR}{signal-to-interference ratio}
\acro{SISO}{Single-Input Single-Output}
\acro{SMV}{single measurement vector}
\acro{SNR}[\textrm{SNR}]{signal-to-noise ratio} 
\acro{sp}[SP]{subspace pursuit}
\acro{SS}{spread spectrum}
\acro{SW}{sync word}
\acro{SAR}{specific absorption rate}
\acro{SSB}{synchronization signal block}
\acro{SR}{shrink and realign}

\acro{tUAV}{tethered Unmanned Aerial Vehicle}
\acro{TBS}{terrestrial base station}

\acro{uUAV}{untethered Unmanned Aerial Vehicle}
\acro{PDF}{probability density functions}

\acro{PL}{path-loss}

\acro{TH}{time-hopping}
\acro{ToA}{time-of-arrival}
\acro{TR}{transmitted-reference}
\acro{TW}{Tracy-Widom}
\acro{TWDT}{TW Distribution Tail}
\acro{TCM}{trellis coded modulation}
\acro{TDD}{Time-Division Duplexing}
\acro{TDMA}{Time Division Multiple Access}
\acro{Tx}{average transmit}

\acro{UAV}{Unmanned Aerial Vehicle}
\acro{uric}[URIC]{upper restricted isometry constant}
\acro{urict}[URICt]{upper restricted isometry constant threshold}
\acro{UWB}{ultrawide band}
\acro{UWBcap}[UWB]{Ultrawide band}   
\acro{URLLC}{Ultra Reliable Low Latency Communications}
         
\acro{wuric}[UWRIC]{upper weak restricted isometry constant}
\acro{wurict}[UWRICt]{UWRIC threshold}                
\acro{UE}{User Equipment}
\acro{UL}{UpLink}

\acro{WiM}[WiM]{weigh-in-motion}
\acro{WLAN}{wireless local area network}
\acro{wm}[WM]{Wishart matrix}                               
\acroplural{wm}[WM]{Wishart matrices}
\acro{WMAN}{wireless metropolitan area network}
\acro{WPAN}{wireless personal area network}
\acro{wric}[WRIC]{weak restricted isometry constant}
\acro{wrict}[WRICt]{weak restricted isometry constant thresholds}
\acro{wrip}[WRIP]{weak restricted isometry property}
\acro{WSN}{wireless sensor network}                        
\acro{WSS}{Wide-Sense Stationary}
\acro{WHO}{World Health Organization}
\acro{Wi-Fi}{Wireless Fidelity}

\acro{sss}[SpaSoSEnc]{sparse source syndrome encoding}

\acro{VLC}{Visible Light Communication}
\acro{VPN}{Virtual Private Network} 
\acro{RF}{Radio Frequency}
\acro{FSO}{Free Space Optics}
\acro{IoST}{Internet of Space Things}

\acro{GSM}{Global System for Mobile Communications}
\acro{2G}{Second-generation cellular network}
\acro{3G}{Third-generation cellular network}
\acro{4G}{Fourth-generation cellular network}
\acro{5G}{Fifth-generation cellular network}	
\acro{gNB}{next-generation Node-B Base Station}
\acro{NR}{New Radio}
\acro{UMTS}{Universal Mobile Telecommunications Service}
\acro{LTE}{Long Term Evolution}

\acro{QoS}{Quality of Service}
\end{acronym}
	
\newcommand{\SAR} {\mathrm{SAR}}
\newcommand{\WBSAR} {\mathrm{SAR}_{\mathsf{WB}}}
\newcommand{\gSAR} {\mathrm{SAR}_{10\si{\gram}}}
\newcommand{\Sab} {S_{\mathsf{ab}}}
\newcommand{\Eavg} {E_{\mathsf{avg}}}
\newcommand{\ft}{f_{\textsf{th}}}
\newcommand{\alphatf}{\alpha_{24}}

\title{
Stochastic Geometry-Based Low Latency Routing in Massive LEO Satellite Networks
}
\author{
Ruibo Wang, Mustafa A. Kishk, {\em Member, IEEE} and Mohamed-Slim Alouini, {\em Fellow, IEEE}
\vspace{-4mm}
}
\maketitle

\begin{abstract}
In this paper, the routing in massive low earth orbit (LEO) satellite networks is studied. When the satellite-to-satellite communication distance is limited, we choose different relay satellites to minimize the latency in a constellation at a constant altitude. Firstly, the global optimum solution is obtained in the ideal scenario when there are available satellites at all the ideal locations. Next, we propose a nearest neighbor search algorithm for realistic (non-ideal) scenarios with a limited number of satellites. The proposed algorithm can approach the global optimum solution under an ideal scenario through a finite number of iterations and a tiny range of searches. Compared with other routing strategies, the proposed algorithm shows significant advantages in terms of latency. Furthermore, we provide two approximation techniques that can give tight lower and upper bounds for the latency of the proposed algorithm, respectively. Finally, the relationships between latency and constellation height, satellites' number, and communication distance are investigated.
\end{abstract}


\section{Introduction}
In recent years, we have witnessed the booming development of low earth orbit (LEO) satellite networks. Companies such as SpaceX, Amazon, and OneWeb are accelerating the formation of a network of tens of thousands of LEO satellites \cite{del2019technical}. Since LEO satellite communication has relatively low latency and unique ability to provide seamless global coverage \cite{kodheli2020satellite}, \cite{yaacoub2020key}, part of real-time communication services are being delivered from ground to space \cite{c2014system}. In terms of low latency and ultra-long distance communication, the LEO satellite network has excellent advantages over ground networks and high orbit satellite networks \cite{chaudhry2020free}. In ultra-long distance communication, multiple satellites are used as relays to complete multi-hop routing. How to select the relay satellite to achieve the minimum latency routing becomes one of the challenges \cite{tang2018multipath}, \cite{he2016delay}. 
\par
Different from the traditional planar routing, satellites are distributed on a closed sphere, and the maximum distance between satellites is limited due to earth blockage \cite{al2021analytic}. For a network where the number and location of satellites are constantly changing, it is more challenging to implement routing in the time-varying topology than in the traditional static topology \cite{sun2020routing}. For small LEO satellites, both computing and storage capacity are limited \cite{zhang2021service}. In a massive LEO satellite network, frequent position changes lead to high computational costs. In addition, each satellite collects only the current state of its neighbors in most cases, which means that it is highly demanding for a single satellite to obtain and store global information such as the location of the satellite. However, using only local information can only get the approximate shortest path, which has limited improvement on the whole constellation latency performance \cite{li2019temporal}.
\par
Existing routing schemes provide strategies to address some of the challenges, but they are not suitable for dynamic large-scale satellite constellations. Stochastic geometry provides a powerful mathematical method for routing in massive constellations. The coverage probability of LEO satellite constellation and two-dimensional plane routing have been studied based on stochastic geometry. Based on these studies, we propose an algorithm to solve the routing problem of a dynamic constellation. At the end of this section, the contributions of this paper are described in more detail.

\subsection{Related Work}
Most of the existing LEO satellite routing is based on the store-and-forward mechanism \cite{lu2015complexity}, \cite{lu2019some}, which undoubtedly brings considerable delay. The following algorithms can achieve real-time communication in specific scenarios \cite{li2019temporal}, \cite{wang2019adaptive}, \cite{pan2019opspf}. In \cite{wang2019adaptive}, medium orbit satellites and high orbit satellites are used to collect and exchange global information to find a route with minimum latency for low-orbit satellites. However, due to the increased complexity of the algorithm, this method is only suitable for small-scale networks but not a massive dynamic network. In \cite{pan2019opspf}, the latency is effectively reduced according to the regular motion of the satellite, and there is no need to collect global information and pay the great computational cost. However, the algorithm is only suitable for a specific small network composed of 8 satellites, and the algorithm cannot optimize link latency. Compared to \cite{pan2019opspf}, the algorithm in \cite{li2019temporal} is local optimum and scalable. By dividing the sphere into many grids, the satellite is positioned by the grid. However, the algorithm reaches the square complexity and is only suitable for static topology. In addition, from the global point of view, it is difficult to guarantee the lower bound of the algorithm.
\par
For massive dynamic satellite networks, the main reason the existing routing algorithms can not combine low complexity and global optimization is that they are designed for each satellite's specific constellations and specific behavior. As an effective mathematical tool, stochastic geometry is especially suitable for analyzing network topology from system-level \cite{haenggi2012stochastic}. So far, many methods have been developed to analyze LEO satellite systems based on stochastic geometry. Binomial point process (BPP) is used to model a closed-area network with a finite number of satellites in \cite{talgat2020stochastic} and \cite{ok-1}. \cite{ok-1}, \cite{talgat2020nearest} and \cite{Al-1} give different forms of contact distance distribution, respectively, that is, the distribution of the distance between a reference point and the nearest satellite. Contact distance distribution provides an important theoretical basis for the analysis of this article. 
\par
In addition, there are several of two-dimensional planar routing strategies based on stochastic geometry \cite{stamatiou2010delay}, \cite{dhillon2015wireless}, \cite{farooq2015stochastic}.
Among them, \cite{sasaki2017energy} and \cite{routingimportant} provide the concept of a reliable region, which ensures the routing can always follow the established direction. The concept of routing efficiency is used to measure the maximum gap between the proposed routing strategy and the optimal one \cite{routingimportant}. By sacrificing the optimality of the algorithm, a sub-optimal routing strategy is given on the premise that only local information is available \cite{richter2018optimal}. According to this idea, the optimal routing is derived in an ideal scenario. Then the sub-optimal routing strategy is proposed when only local information is available.

\subsection{Contribution}
So far, this is the first study of satellite routing based on stochastic geometry. The contributions can be summarized as follows:
\begin{itemize}
    \item Three propositions are given in the ideal scenario where there are available satellites at any location. Based on these propositions, we provide a solution for the ideal scenario and use it as an upper bound for the proposed algorithm.
    \item Equal interval, minimum deflection angle, and maximum step size relay strategies are derivatives of propositions in the ideal scenario. We obtain the proposed algorithm by improving the equal-interval relay strategy. The remaining two are used as the baselines.
    \item We provide two approximations to estimate the gap between the algorithm and the best possible solution. Numerical results show that these two approximations can give tight upper and lower bounds for the algorithm delay.
    \item According to three deterministic LEO satellite constellations, algorithm complexity, average and maximum search area required for finding at least one satellite are analyzed.
    \item We study the influence of parameters such as communication distance, constellation height, and the number of satellites on latency.
\end{itemize}

\begin{table*}[]
\centering
\caption{Summary of Notations.}
\begin{tabular}{|M{2.8cm} | M{12cm}|}
\hline
\textbf{Notation}     & \textbf{Description}      \\ \hline  \hline
$N_{\rm{Sat}}$; $n$; $ N_{\min}$ &
  Number of satellites; number of hops that one link contains; minimum number of hops \\ \hline
${r_\oplus}$; $r_{\rm{Sat}}$; $r$    & Radius of the Earth; height of the satellite orbits; radius of the sphere where satellites locate  \\ \hline
$\mathcal{H}$; $h_i$; $d_i$   & The set which contains the IDs of a link; ID of the $i^{th}$ satellite; distance of the $i^{th}$ hop  \\ \hline
$x_i$, $\theta_i$, $\varphi_i$; $\Phi$  & The location, polar angle, azimuth angle of the $i^{th}$ satellite; the homogeneous BPP  \\ \hline
$T$; $\varepsilon$  & Latency of the multi-hop link;  link tolerable probability of interruption    \\ \hline
$\theta_i^h$; $\theta_{_{02n}}^h$ &
Dome angle between satellites $x_{h_{i-1}}$ and $x_{h_i}$; starting satellite $x_{h_0}$ and ending satellite $x_{h_n}$ \\ \hline
$d_{\max}$; ${\theta_{\max}} $    & Maximum communication distance; upper bound of dome angle between satellites \\ \hline
$\theta_0$; $\theta_r$ & Contact angle; reliable angle      \\ \hline
$\widetilde{E}_1$; $\widetilde{E}_2$ & Contour integral approximation of the efficiency; binomial approximation of the efficiency    \\ \hline
\end{tabular} 
\end{table*}

\section{Optimal Routing Scheme}
Let us consider a scenario where two satellites are too far apart to communicate directly. Several satellites act as relays to complete multi-hop satellite to satellite link communication.

\subsection{Problem Formulation}
To formalize the problem, this section introduces ({\romannumeral1}) satellite distribution, ({\romannumeral2}) link routing model, ({\romannumeral3}) coordinate system and ({\romannumeral4}) optimization problem in order.
\par
Consider a massive constellation composed of $N_{\rm{Sat}}$ satellites, which are independently distributed on a spherical surface according to a homogeneous Binomial Point Process (BPP) \cite{ok-1}. The radius of the sphere is denoted as $r={r_\oplus}+r_{\rm{Sat}}$, where $r_\oplus=6371\rm{ \,Km}$ is the radius of the Earth, and $r_{\rm{Sat}}$ is the height of the satellite orbits. 
\par
The latency required for transmission is often measured in milliseconds, which is much smaller than the orbital period of LEO satellites. The change of satellite position with time in single routing is negligible. A transmission from one satellite to another is called a hop. A link with $n$ hops can be expressed as $\mathcal{H}=\{h_0,h_1,...,h_n\}$. $h_i$ is the ID of the $i^{th}$ satellite, which is a positive integer less than $N_{\rm{Sat}}$. $x_{h_0}$ and $x_{h_n}$ are the positions of the starting point and the ending point, respectively. 
\par
\begin{figure*}[t]
	\centering
	\includegraphics[width=0.98\linewidth]{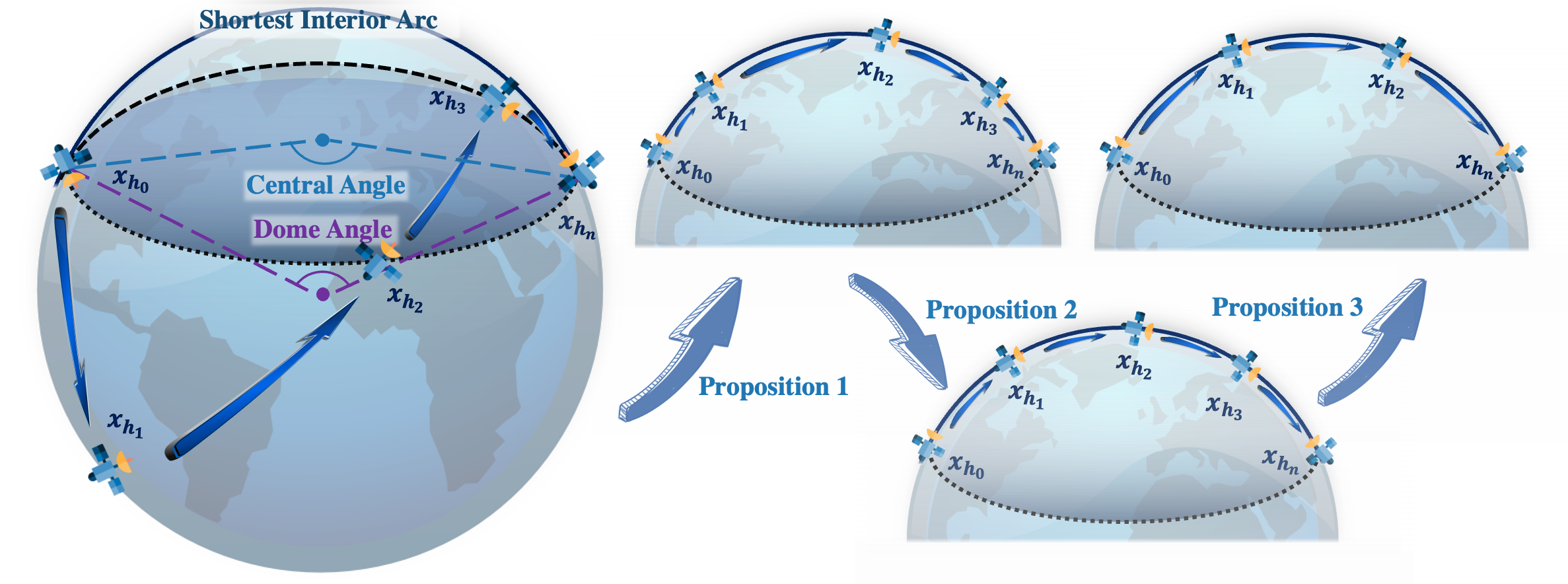}
	\caption{Explanatory figure of the three propositions.}
	\label{fig:System model}
	\vspace{-0.4cm}
\end{figure*}

Since the distribution of satellites forms a homogeneous BPP, the rotation of the coordinate system do not affect the distribution. Set the center of the Earth as the origin. All satellites have the same radial distance $r$. We establish the coordinate system by the coordinates of the starting satellite $x_{h_0}$ and the ending satellite $x_{h_n}$ of the multi-hop link. As is shown in Fig.~\ref{fig:System model}, the $x$-axis is parallel to the line segment between $x_{h_0}$ and $x_{h_n}$, and the $z$-axis is the midperpendicular of this segment, so the $y$-coordinates of $x_{h_0}$ and $x_{h_n}$ are 0. Since satellites are distributed on a sphere, spherical coordinates are more practical than rectangular coordinates. Coordinate $(r,\theta_i,\varphi_i)$ is used to represent the location of $i^{th}$ satellite $x_i$. $\theta_i$ and $\varphi_i$ are the polar and azimuth angles, respectively. Furthermore, the homogeneous BPP is denoted as $\Phi=\{x_1,x_2,...,x_{N_{\rm{Sat}}}\}$. $d_i$ is used to describe the distance of the $i^{th}$ hop, that is, the spatial distance from $x_{h_{i-1}}$ to $x_{h_i}$,
\begin{equation}
\begin{split}
\label{d_i}
    d_i = r \bigg[ 2\Big(1 - \cos{\theta_{h_{i-1}}}\cos{\theta_{h_i}}
    -\sin{\theta_{h_{i-1}}}\sin{\theta_{h_i}}\cos\left(\varphi_{h_{i-1}}-\varphi_{h_i}\right)\Big)\bigg]^{\frac{1}{2}},
\end{split}
\end{equation}
where $i=1,2,...,n$.
\par
To minimize the latency by selecting the number of satellites and their positions, we consider the following optimization problem,
\begin{subequations} 
	\begin{alignat}{2}
		\mathscr{P}_0:\quad &\underset{ {n,\mathcal{H}}}{\text{minimize}}        &\quad& T = \frac{1}{c}\sum_{i=1}^{n} d_i, \label{eq:opt}\\
		&\textrm{subject to:}    &      & d_i\leq 2\sqrt{r^2-r_\oplus^2}, \; \forall i, \label{st:constraint} \\
		&      &      &   d_i\leq d_{\max}, \; \forall i. \label{st:constraint2}
	\end{alignat}
\end{subequations}
\par
In (\ref{eq:opt}), the optimization objective is the latency of the multi-hop link, where $c=3\times10^2 \rm{ \,Km/ms}$ is the speed of laser propagation. Constraint (\ref{st:constraint}) guarantees that the satellites are within line-of-sight of each other \cite{ok-1}, and constraint (\ref{st:constraint2}) limits the maximum communication distance $d_{\max}$ between satellites. Note that We omit the power constraint issues in $\mathscr{P}_0$. Since the objective function is related to the position and number of satellites, the problem is not convex.

\subsection{The Ideal Scenario Solution}\label{The Ideal Scenario Solution}
To make the problem more manageable, we start with an ideal scenario, which assumes satellites are available anywhere on the sphere. Before solving the optimization problem $\mathscr{P}_0$, the following definitions are required.
\begin{definition}[Central Angle]
For a circle passing satellites A and B, the central angle of the circle is the angle between the line connecting A and the center of the circle and the line connecting B and the center of the circle.
\end{definition}
\begin{definition}[Dome Angle]
For a circle centered at the origin, passing satellites A and B, the central angle for this specific circle is called the dome angle.
\end{definition}
\begin{definition}[Shortest Inferior Arc]
The circle centered at the origin, with radius $r$, passing the staring point $x_{h_0}$ and the ending point $x_{h_n}$, are divided into two arcs by $x_{h_0}$ and $x_{h_n}$. The arc with a shorter arc length is called the shortest inferior arc.
\end{definition}
\par
An ideal solution of problem $\mathscr{P}_0$ is derived through the following three propositions.
\par

\begin{proposition}\label{prop1} 
In the ideal scenario, optimal positions $x_{h_i}^*$ in $\mathscr{P}_0$ are located on the shortest inferior arc.
\end{proposition}
\begin{proof}
See Appendix~\ref{app:prop1}.
\end{proof}

Based on proposition~\ref{prop1}, all satellites are assumed to locate on the shortest inferior arc. Therefore, an equivalence problem for $\mathscr{P}_0$ is given by, 
\begin{subequations} 
	\begin{alignat}{2}
		\mathscr{P}_1:\quad &\underset{ {n,\mathcal{H}}}{\text{minimize}}        &\quad& T = \frac{1}{c}\sum_{i=1}^{n} 2r\sin\left(\frac{\theta_i^h}{2}\right),\label{opt2}\\
		&\textrm{subject to:}    &      & \theta_i^h\leq 2\arccos\left(\frac{r_\oplus}{r}\right) \; \forall i, \label{st:constraint2-1}\\
		&                  &      &   \theta_i^h\leq 2\arcsin\left(\frac{d_{\max}}{r}\right), \; \forall i,\label{st:constraint2-2}\\
		&                  &      &   \sum_{i=1}^{n} \theta_i^h = \theta_{_{02n}}^h,\label{st:constraint2-3}
	\end{alignat}
\end{subequations}
where $\theta_i^h$ is the dome angle between satellites $x_{h_{i-1}}$ and $x_{h_i}$. As is shown in Fig.~\ref{fig:angles} $\theta_{_{02n}}^h$ is the dome angle between starting satellite $x_{h_0}$ and ending satellite $x_{h_n}$, which is given as,
\begin{equation}\label{theta_02n}
\begin{split}
    \theta_{_{02n}}^h = \arcsin \bigg(   \frac{\sqrt{2}}{2}\Big(1 - \cos{\theta_{h_0}}\cos{\theta_{h_n}}
    -\sin{\theta_{h_0}}\sin{\theta_{h_n}}\cos\left(\varphi_{h_0}-\varphi_{h_n}\right)\Big)^{\frac{1}{2}} \bigg).
\end{split}
\end{equation}
${\theta_{_{02n}}^h}$ is also defined as the dome angle of the multi-hop link. It can be derived intuitively by the formula (\ref{d_i}) with the aid of simple geometric relations. The following proposition will further give a more specific distribution of relay satellite positions.

\begin{figure}[t]
	\centering
	\includegraphics[width=0.7\linewidth]{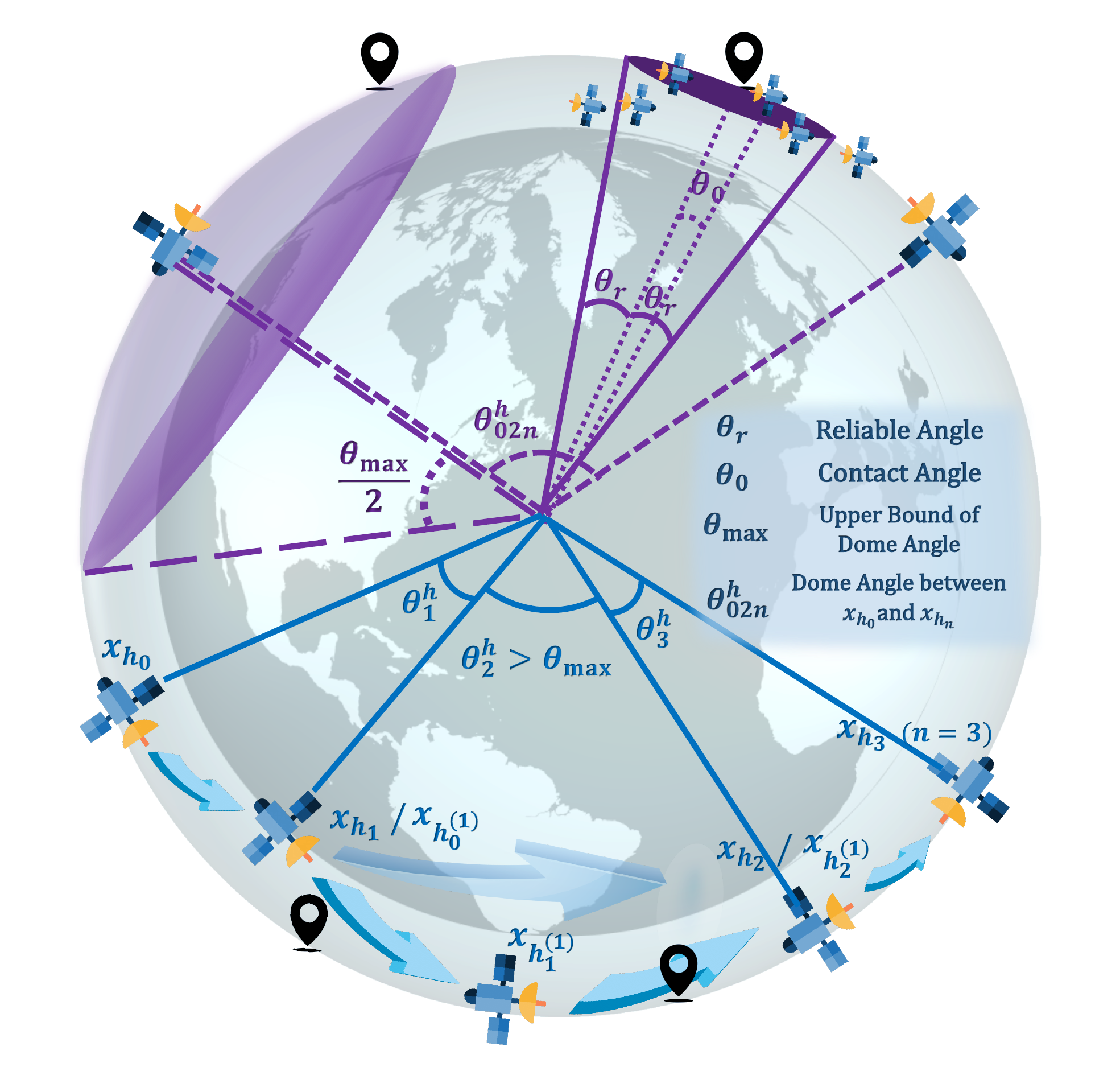}
	\caption{Schematic diagram of equal interval routing strategy.}
	\label{fig:angles}
	\vspace{-0.4cm}
\end{figure}

\begin{proposition}\label{prop2} 
In the ideal scenario, for an $n$-hop link, if the satellites are located on the shortest inferior arc, the optimal dome angle $\theta_i^{h*}$ in $\mathscr{P}_1$ is equal to ${\theta_{_{02n}}^h}/{n}$. 
\end{proposition}
\begin{proof}
See Appendix~\ref{app:prop2}.
\end{proof}
Proposition~\ref{prop2} decreases delay by the equally spaced distribution of relay satellites, while proposition~\ref{prop3} minimizes the latency by determining the optimal number of satellites. Both propositions are shown in Fig.~\ref{fig:System model}.

\begin{proposition}\label{prop3} 
In an ideal scenario, assume the satellites are equally spaced distributed on the shortest inferior arc, the optimal number of hops is
\begin{equation}\label{N_min}
    N_{\min}=\bigg\lceil \frac{\theta_{_{02n}}^h}{\theta_{\max}} \bigg\rceil + 1,
\end{equation}
where $\lceil . \rceil$ means rounding up to an integer, and 
\begin{equation}\label{theta_max}
    \theta_{\max} = \min\left\{2\arccos\left(\frac{r_\oplus}{r}\right),2\arcsin\left(\frac{d_{\max}}{r}\right)\right\}.
\end{equation}
\begin{proof}
See Appendix~\ref{app:prop3}.
\end{proof}
\end{proposition}

In proposition~\ref{prop3}, $\theta_{\max}$ is the upper bound of the dome angle between satellites that have established communication links. $\theta_{\max}$ ensures that two satellites are within the LoS region of each other and the maximum communication distance $d_{\max}$. By combining the above propositions, the global optimum solution to the problem $\mathscr{P}_0$ under ideal conditions is given by the following theorem.

\begin{theorem}\label{ideal solution}
In the ideal scenario, the global optimal multi-hop link in $\mathscr{P}_0$ has $N_{\min}$ hops, and each hop is located on the inferior arc with equal interval distribution, and the dome angle between each hop is ${\theta_{_{02n}}^h}/{N_{\min}}$.
\end{theorem}

\subsection{Practical Strategies Discussion}
Although the optimal solution is derived in section~\ref{The Ideal Scenario Solution}, it cannot be applied in practice because an infinite number of satellites is required. Based on propositions 1\,-\,3, we designed three strategies to transition multi-hop routing from the ideal scenario to the situation with limited satellites. Fig.~\ref{fig:Examples of four algorithms} is a top view along the direction of the negative $z$-axis. It gives an example of these strategies. 
\par
In \textbf{minimum deflection angle strategy}, each satellite should look for the satellite with the least deflection from the shortest inferior arc as its next hop. Only satellites satisfying the distance constraints are eligible to be relay satellites. The next-hop satellite also needs to be shorter from the ending satellite than the previous one to ensure that each hop keeps approaching the destination satellite. These requirements also need to be met in the two subsequent strategies. From the algorithm's perspective, since $\theta=0$ for the shortest inferior arc, the strategy finds the satellite with the minimum value of $\theta_i$ that meets the requirements.
\par
\textbf{Equal interval strategy} finds the nearest satellite as the relay in every optimal position obtained under the ideal scenario. As an intuitive extension of the ideal scenario solution, this strategy can bring extremely low latency. The cost of low delay is the poor reliability since it is highly likely that relay satellites do not meet the constraints (\ref{st:constraint}) and (\ref{st:constraint2}).
\par
In \textbf{maximum stepsize strategy}, the satellite chooses the farthest satellite within communication range as its next hop. It reduces the number of hops as much as possible on the premise of ensuring successful communication. In order to avoid the relay satellite being too far away from the shortest inferior arc, we set up a reliable region, which is the dark area in Fig.~\ref{fig:Examples of four algorithms}.

\begin{figure*}[t]
	\centering
	\includegraphics[width=0.9\linewidth]{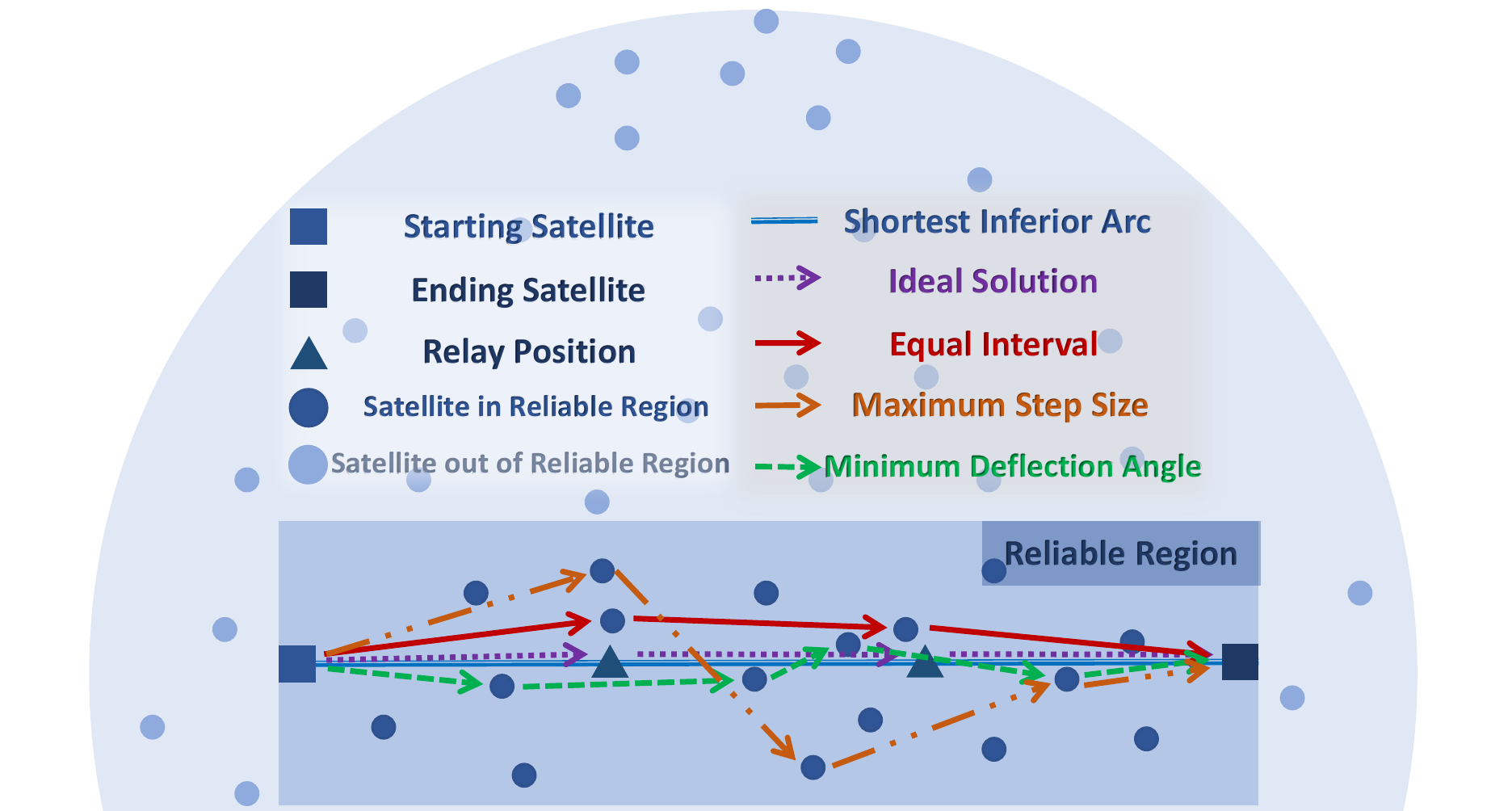}
	\caption{An example of three strategies.}
	\label{fig:Examples of four algorithms}
	\vspace{-0.4cm}
\end{figure*}

As a result, minimum deflection angle strategy and maximum stepsize strategy are set as baselines. The proposed algorithm is designed on the basis of the equal interval strategy, and it is proved to have the lowest latency and high reliability.

\section{Algorithm Design and Performance Analysis}\label{The Practical Solution}
In this section, we first determine the number of hops of the multi-hop link by introducing contact angle and reliable angle. After that, a complete nearest neighbor search algorithm is given, and its reliability is analyzed. Finally, we define link efficiency to measure the maximum gap between algorithm delay and possible optimal solution.

\subsection{Contact Angle and Reliable Angle}
Since the interval ${\theta_{_{02n}}^h}/{n}$ decreases as the number of hops $n$ increases, one way to improve the reliability of equal interval strategy is to increase $n$. However, proposition~\ref{prop3} shows that the latency is also increased with $n$. To choose a proper $n$ which can balance the latency and reliability, the concepts of contact angle $\theta_0$ and reliable angle $\theta_r$ need to be introduced first, which is shown at the top of Fig.~\ref{fig:angles}.
\begin{definition}[Contact Angle]
The contact angle is the dome angle between a randomly placed reference and the closest point from the process (the nearest satellite in this article).
\end{definition}
\par
Since the satellites form a uniform BPP, any randomly selected reference points have the same contact angle distribution.

\begin{lemma}\label{CDF of contact}
The \ac{CDF} of the contact angle is obtained as,
\begin{equation}
    F_{\theta_0}\left(\theta\right) = 1 - \left( \frac{1+\cos \theta}{2} \right) ^ {N_{\rm{Sat}}}, \ 0 \leq \theta \leq \theta_{\max},
\end{equation}
where $\theta_{\max}$ is defined in (\ref{theta_max}).
\begin{proof}
See Appendix~\ref{app:CDF of contact}.
\end{proof}
\end{lemma}
Based on Lemma~\ref{CDF of contact}, the \ac{PDF} of the contact angle can be obtained by taking the derivative of \ac{CDF} with respect to $\theta$.
\begin{lemma}\label{PDF of contact}
The \ac{PDF} of the contact angle is obtained as,
\begin{equation}
    f_{\theta_0}\left(\theta\right) = \frac{N_{\rm{Sat}}
    }{2} \sin{\theta} \left( \frac{1+\cos \theta}{2} \right) ^ {N_{\rm{Sat}}-1}, \ 0 \leq \theta \leq \theta_{\max},
\end{equation}
where $\theta_{\max}$ is defined in (\ref{theta_max}).
\end{lemma}

\begin{definition}[Reliable Angle]
Reliable angle is the minimum dome angle that ensures that at least one satellite can be found within a specified range.
\end{definition}

However, even given a large region for search, no satellite may be available because of the randomness. Therefore, we can only guarantee that the probability of not finding any satellite is lower than an acceptable threshold. The value of reliable angle is related to this predefined threshold.

\begin{definition}[Link Tolerable Probability of Interruption]
Link tolerable probability of interruption $\varepsilon$ is the upper bound of the probability that no satellite is available within the reliable angle range in at least one hop.
\end{definition}
The absence of a satellite available within a reliable angle range does not mean that the hop will be interrupted because the interruption also depends on the location of the other relay satellite. Therefore, $\varepsilon$ is not equivalent to the average link interruption probability but an upper bound. In addition, $\varepsilon$ can be regarded as a system parameter determined by the requirements rather than an optimization variable. For a fixed $\varepsilon$, the more hops the link has, the higher the reliability required for a single hop. Therefore, reliable angle $\theta_r$ is a monotonically increasing function of $n$. The following lemma will give the relationship among reliable angle $\theta_r$, link tolerable probability of interruption $\varepsilon$, and the number of hops $n$.
\begin{lemma}\label{reliable angle with n}
For an $n$-hop link with link tolerable probability of interruption $\varepsilon$, the reliable angle $\theta_r$ is given by,
\begin{equation}
    \theta_r\left(n\right) = \arccos \left( 2 \left( 1 - \left( 1 - \varepsilon \right)^{\frac{1}{n}} \right)^{\frac{1}{N_{\rm{Sat}}}} - 1 \right).
\end{equation}
\begin{proof}
See Appendix~\ref{app:reliable angle with n}.
\end{proof}
\end{lemma}

\subsection{Type-\uppercase\expandafter{\romannumeral1} Interruption Analysis}
Through the above analysis, the following results about the number of hops $n$ can be summarized. The latency increases monotonically with the increase of $n$. The relationship between interruption probability and the number of hops is not intuitive. Increasing $n$ requires a lower interruption probability for a single hop but brings a larger area for finding a satellite. If $n$ is too large and the single-hop interval is too small, two relay locations of the one-hop may choose the same satellite, which leads to severe errors. Furthermore, on the premise that the  probability of type-\uppercase\expandafter{\romannumeral2} interruption is lower than $\varepsilon$, the number of hops should be as small as possible. To satisfy the distance constraints, the dome angle of each hop $\theta_i^{h}$ should satisfy,
\begin{equation}
    \theta_i^{h} + 2 \theta_r\left( n \right) \leq \theta_{\max}.
\end{equation}
To ensure that the multi-hop communication can be completed within $n$ hops, we have, 
\begin{equation}
    n \cdot \theta_i^{h} \geq \theta_{_{02n}}^h.
\end{equation}
By combining the above two inequalities, a loose lower bound on $\theta_r\left( n \right)$ can be obtained,
\begin{equation}
    n \left( \theta_{\max} - 2 \theta_r\left( n \right) \right) \geq \theta_{_{02n}}^h,
\end{equation}
To avoid the possibility of selecting the same satellite for two relay positions of a single hop, an upper bound of $\theta_r\left( n \right)$ is given as,
\begin{equation}
    \theta_r\left( n \right) \leq \frac{\theta_{\max}}{2},
\end{equation}
The following algorithm can give the minimum number of hops between the upper and lower bounds through iteration. 
 \begin{algorithm}[!ht] 
	\caption{Iterative Method for Deriving the Number of Hops}
	\label{alg.number of hops} 
	\begin{algorithmic} [1]
		
		\STATE \textbf{Input}: Dome angle $\theta_{_{02n}}^h$, number of satellites $N_{\rm{Sat}}$ and link tolerable probability of interruption $\varepsilon$.
		
		\STATE $n\leftarrow N_{\min}$.
		
		\STATE $\theta_r \leftarrow \arccos \left( 2 \left( 1 - \left( 1 - \varepsilon \right)^{\frac{1}{n}} \right)^{\frac{1}{N_{\rm{Sat}}}} - 1 \right)$.
		
		\WHILE{$\frac{1}{2} \left( \theta_{\max} -  \frac{\theta_{_{02n}}^h}{n} \right) \leq \theta_r \leq \frac{1}{2} \theta_{\max}$}
		\STATE $n \leftarrow n + 1$.
		\STATE $\theta_r \leftarrow \arccos \left( 2 \left( 1 - \left( 1 - \varepsilon \right)^{\frac{1}{n}} \right)^{\frac{1}{N_{\rm{Sat}}}} - 1 \right)$.
		\ENDWHILE
		
		\STATE \textbf{Output}: Minimum number of hops $n$ and reliable angle $\theta_r$.
	\end{algorithmic}
\end{algorithm}	
\par
Note that the minimum number of hops is not related to the positions of satellites and it can be expressed as,
\begin{equation}
\begin{split}
    &\hat{N}_{\min} = \min\Bigg\{ n: \left( \theta_{\max} - 2\theta_r \left(n-1\right) \right)\left( \theta_{\max} - 2\theta_r \left(n\right) \right) < 0 \ \ \mathbf{or} \\
    & \left( \theta_{\max} - \frac{\theta_{_{02n}}^h}{n-1}  - 2\theta_r\left(n-1\right)  \right)\left( \theta_{\max} -  \frac{\theta_{_{02n}}^h}{n}  - 2\theta_r \left(n\right) \right) < 0 \Bigg\},
\end{split}
\end{equation}
which is another representation of step (4) of the algorithm. Both $\theta_r\left(n\right)$ and $\frac{1}{2} \left( \theta_{\max} -  {\theta_{_{02n}}^h}/{n} \right)$ increase with $n$. When the algorithm ends the loop as $\frac{1}{2} \left( \theta_{\max} -  {\theta_{_{02n}}^h}/{n} \right) > \theta_r\left(n\right)$ satisfied, the output $n$ is the required minimum number of hops. Otherwise, when the algorithm ends the loop as $2\theta_r\left( n \right) > \theta_{\max}$, no value of $n$ guarantees tolerable probability of interruption $\varepsilon$. For a constellation with a small number of satellites, it is not realistic to guarantee a low tolerable probability of interruption. Such problems due to poor system design are defined as type-\uppercase\expandafter{\romannumeral1} interruption.
\par
\begin{definition}[Type-\uppercase\expandafter{\romannumeral1} interruption]
Type-\uppercase\expandafter{\romannumeral1} interruption is a qualitative indicator to describe the rationality of multi-hop communication system design.
\end{definition}
\par

In addition to running an algorithm to determine whether the type-\uppercase\expandafter{\romannumeral1} interruption occurred, The proposition also provides a sufficient condition for the type-\uppercase\expandafter{\romannumeral1} interruption not to occur.
\begin{proposition}\label{Minimum Nsat}
If there exists a $\theta_{t}<\frac{1}{2}\theta_{\max}$ let the following inequality satisfied, there must be a routing scheme that makes the probability that no satellite is available within the reliable angle range in at least one hop lower than $\varepsilon$,
\begin{sequation}\label{sufficient condition}
    N_{\rm{Sat}} \geq \frac{1}{\ln \left( \frac{1 + \cos \theta_t}{2} \right)} \ln \left( 1 - \left( 1 - \varepsilon \right)^{1 \big/ \left( \Big\lceil  \frac{\theta_{_{02n}}^h}{\theta_{\max}-2\theta_t} \Big\rceil + 1\right) }  \right),
\end{sequation}
where $\theta_{_{02n}}^h$ and $\theta_{\max}$ are defined in (\ref{theta_02n}) and (\ref{theta_max}) respectively.
\begin{proof}
See Appendix~\ref{app:Minimum Nsat}.
\end{proof}
\end{proposition}

\subsection{Type-\uppercase\expandafter{\romannumeral2} Interruption Analysis and Nearest Neighbor Search Algorithm}
Considering that even if the constellation is suitable for multi-hop transmission, communication interruption may still happen due to the randomness of the satellite position. Such interruptions are defined as type-\uppercase\expandafter{\romannumeral2} interruption.

\begin{definition}[Type-\uppercase\expandafter{\romannumeral2} Interruption]
Type-\uppercase\expandafter{\romannumeral2} interruption is an event that happens when the distance in any hop does not satisfy at least one constraint in $\mathscr{P}_0$.
\end{definition}
\par

Although the two types of interruptions happen for different reasons, the occurrences of these two types of interruptions are not independent. The occurrence of type-\uppercase\expandafter{\romannumeral1} interruption often leads to type-\uppercase\expandafter{\romannumeral2} interruption. Because type-\uppercase\expandafter{\romannumeral2} interruption cannot be avoided by the parametric design of the satellite constellation, so we deal with the interruption after it occurs. Suppose the distance between each satellite is too far to communicate. In that case, the satellite at the starting point of the hop is expected to looking for one or several satellites closest to the shortest inferior arc as relays. As is shown at the bottom of Fig.\ref{fig:angles}, it can be regarded as using minimum deflection angle strategy within a single hop.
\par

As mentioned, the algorithm proposed in this article is based on the equal interval strategy. If two types of interruptions are resolved, the probability of errors occurring in the equal-interval strategy is significantly reduced, thus ensuring low latency and high reliability. The practical nearest neighbor search algorithm is divided into four stages: (\romannumeral1) calculate the minimum number of hops through iteration, (\romannumeral2) find the relay position according to equal interval strategy, (\romannumeral3) find nearest satellite in the neighborhood of the relay position to establish the link, and (\romannumeral4) adopting minimum deflection angle strategy in the single hop when the two satellites of the hop cannot satisfy the distance constraints. The last three steps of the algorithm are as follows.

 \begin{algorithm}[!ht]  
	\caption{Nearest Neighbor Search Algorithm}
	\label{alg.Nearest Neighbor Search Algorithm} 
	\begin{algorithmic} [1]
		
		\STATE \textbf{Input}: Locations of point process $\Phi$, the number of hops $\hat{N}_{\min}$, starting point ID $h_0$ and ending point ID $h_{\hat{N}_{\min}}$.
        
        \STATE \textbf{Initialize}: $T \leftarrow 0$.
        
		\FOR{$i = 1 : \hat{N}_{\min}-1$}
		\STATE $\theta_i^h \leftarrow \theta_{h_0} \Big| \frac{2i}{\hat{N}_{\min}} - 1 \Big|$.
		\IF{$i < \frac{\hat{N}_{\min}}{2}+1$}
		\STATE $\varphi_i^h \leftarrow 0$.
		\ELSE 
		\STATE $\varphi_i^h \leftarrow \pi$.
		\ENDIF
		\STATE ${h_i} \leftarrow \arg {\min\limits _{j}} \; d\big( \theta_i^h,\varphi_i^h,\theta_j,\varphi_j \big)$.
		\ENDFOR
		
		\STATE $\mathcal{H} \leftarrow \{ h_0,h_1,...,h_{\hat{N}_{\min}-1},h_{\hat{N}_{\min}} \}. $
		
		\FOR{$i = 1 : \hat{N}_{\min}$}
		\IF{$d( \theta_{h_{i-1}}, \varphi_{h_{i-1}}, \theta_{h_i}, \varphi_{h_i} ) > d_{\max}$}
		\STATE Use minimum deflection angle strategy to find the relay satellite IDs  $\mathcal{H}^{(i)}\{h_0^{(i)},h_1^{(i)}...,h_n^{(i)}\}$ in $i^{th}$-hop.
		\STATE $T \leftarrow T + \sum_{j=1}^{n}
		d \Big( \theta_{h_{j-1}^{(i)}},\varphi_{h_{j-1}^{(i)}},\theta_{h_{j}^{(i)}},\varphi_{h_{j}^{(i)}} \Big).$
		\ELSE 
		\STATE $T \leftarrow T + d( \theta_{h_{i-1}},\varphi_{h_{i-1}},\theta_{h_i},\varphi_{h_i} ).$
		\ENDIF
		\ENDFOR
		
		\STATE \textbf{Output}: IDs of the multi-hop link $\mathcal{H}$ and Latency $T$.
	\end{algorithmic}
\end{algorithm}	
To simplify the description of the algorithm, the distance between two points is defined as,
\begin{equation}\label{d function}
\begin{split}
    d\left( \theta_1, \varphi_1, \theta_2, \varphi_2 \right) = r \big( 2 ( 1 - \cos{\theta_1}\cos{\theta_2} 
    - \sin{\theta_1}\sin{\theta_2}\cos\left(\varphi_1-\varphi_2\right) )\big)^\frac{1}{2}.
\end{split}
\end{equation}
In addition, the start ID $h_0^{(i)} = h_{i-1}$ and the end ID $h_n^{(i)} = h_i$ in set $\mathcal{H}^{(i)}$. The nearest neighbor search algorithm cannot guarantee finding the optimal solution even when the two types of interruptions do not occur. For example, a link with many hops may meet the distance constraints even after two links are merged. Since the sum of the two sides of the triangle is greater than the third, the combined link has a lower latency. Therefore, it is necessary to analyze the latency performance of the algorithm.

\subsection{Efficiency Analysis}\label{efficiency}
For an optimization problem that is difficult to find the optimal solution, the most concerning issue is the gap between the found solution and the optimal solution. Unfortunately, according to the available data, no algorithm can find the optimal solution to the problem. The latency of the optimal solution in the ideal scenario is an unattainable lower bound. It is also an upper bound of the difference between the proposed method and the optimal solution. Therefore, efficiency is defined to quantify the difference.

\begin{definition}[Efficiency]
Efficiency is the ratio of minimum latency in the ideal scenario to the latency of the proposed method.
\end{definition}

Since satellites are uniformly and independently distributed on the sphere and the intervals of relay positions on multi-hop links are equal, the distance between single hops is independent and identically distributed. Therefore, analyzing the efficiency of multi-hop links can be equivalent to studying that of single-hop. The increase in the distance caused by random distribution can be equivalent to the increase of the dome angle. In other words, it offsets the random distribution of satellites by moving their relay positions. Thus, the following two approximations are given.

\begin{theorem}\label{contour integral approximation}
For an $\hat{N}_{\min}$-hop link with dome angle ${\theta_{_{02n}}^h}$, the contour integral approximation of the efficiency is given as,
\begin{equation}
    \widetilde{E}_1 =  \frac{N_{\min}\cdot\sin\left(\frac{\theta_{_{02n}}^h}{2N_{\min}}\right)} {\hat{N}_{\min}\cdot\sin\left(\frac{\theta_{_{02n}}^h}{2\hat{N}_{\min}}\right)\left( 2\overline{\alpha}\left(\frac{\theta_{_{02n}}^h}{2\hat{N}_{\min}}\right)-1 \right)},
\end{equation}
where $N_{\min}$ is defined in (\ref{N_min}), and $\overline{\alpha}\left(\theta^h\right)$ is defined as,
\begin{equation}\label{overline alpha}
\begin{split}
    \overline{\alpha}\left(\theta^h\right) = \frac{\sqrt{2}}{2\pi} \int_0^{\theta_{\max}} \int_0^{\pi} \frac{f_{\theta_0}\left(\theta\right)}{\sin\left(\frac{\theta^h}{2}\right)} \Big( -\cos(\theta_0)\cos(\theta^h)
    -\sin(\theta_0)\sin(\theta^h)\cos\varphi+1 \Big)^{\frac{1}{2}} \,\mathrm{d}\varphi \,\mathrm{d}\theta.
\end{split}
\end{equation}

\begin{proof}
See Appendix~\ref{app:contour integral approximation}.
\end{proof}
\end{theorem}

\begin{theorem}
For an $\hat{N}_{\min}$-hop link with dome angle ${\theta_{_{02n}}^h}$, the binomial approximation of the efficiency is given as,
\begin{equation}
    \widetilde{E}_2 =  \frac{N_{\min}\cdot\sin\left(\frac{\theta_{_{02n}}^h}{2N_{\min}}\right)} {\hat{N}_{\min}\cdot\eta\left(\frac{\theta_{_{02n}}^h}{2\hat{N}_{\min}}\right)},
\end{equation}
where $N_{\min}$ is defined in (\ref{N_min}), and $\eta\left(\theta^h\right)$ is defined as,
\begin{equation}
\begin{split}
    \eta\left(\theta^h\right) &= \int_0^{\theta_{\max}} \int_0^{\theta_{\max}} \frac{1}{4} f_{\theta_0}\left(\theta_1\right) f_{\theta_0}\left(\theta_2\right) 
    \Big( \sin(\theta^h-\theta_1-\theta_2)+\sin(\theta^h+\theta_1-\theta_2) \\
    &+\sin(\theta^h-\theta_1+\theta_2)+\sin(\theta^h+\theta_1+\theta_2) \Big)  \,\mathrm{d}\theta_1\,\mathrm{d}\theta_2.
\end{split}
\end{equation}
\begin{proof}
Assuming that the contact angle between the relay position and its nearest satellite is $\theta_0$, the satellites are uniformly distributed on a circle with radius $r\sin\theta_0$. By approximating this distribution as a binomial distribution,  satellites are distributed at the nearest or farthest from the adjacent relay position with equal probability. Take the expectation of contact angles, and the above result can be obtained.
\end{proof}
\end{theorem}

\begin{table*}[]
\caption{Reliability analysis of deterministic LEO satellite constellations.}
\renewcommand\arraystretch{1}
 \resizebox{16.2cm}{3.15cm}{
\begin{tabular}{|c|c|c|c|}
\hline

 & Starlink      & OneWeb        & Kuiper          \\ \hline  \hline
Constellation altitude [Km]           & 550           & 1200     & 590,   610, 630 \\ \hline
Number   of (planned) satellites        & 11927         & 650           & 3236            \\ \hline
Expectation of contact angle               &0.0162           & 0.0695          & 0.0312           \\ \hline
Number   of hops                        & 9~/~10          & 69~/~8          & 12~/~13           \\ \hline
Reliable   angle [rad]                  & 0.0386~/~0.0481 & 0.1996~/~0.2026 & $\approx0.0765~/~\approx0.0941$   \\ \hline
Minimum number of satellites required & 710~/~2535     & 2053~/~7889    & $ \approx 777~/~\approx2520 $     \\ \hline
Type-\uppercase\expandafter{\romannumeral1} interruption occurs or not        & No~/~No         & Yes~/~Yes        & No~/~No           \\ \hline
Probability of Type-\uppercase\expandafter{\romannumeral2} interruption occurs &   $<0.01\%~/~<0.01\%$  & $9.41\%~/~100\%$   &  $<0.01\%~/~<0.01\%$   \\ \hline
Efficiency &  $99.44\%~/~99.17\%$  & $97.80\%~/~96.27\%$   &  $97.91\%~/~97.56\%$   \\ \hline
\end{tabular}
    }
\label{table:models}
\end{table*}

\section{Numerical Results}
This section analyzes the performance of the algorithm based on the results of numerical simulation. For the existing deterministic constellations, we analyze the feasibility of the algorithm. Then different approximation methods and routing strategies are compared from the perspective of latency.

\subsection{Reliability Analysis of Deterministic Constellations}
Table~\ref{table:models} shows the simulation results of three deterministic LEO satellite constellations \cite{robert2020small}. Set the maximum distance at which the satellite can maintain stable communication as $d_{max}=3000\rm{ \,Km}$. Within this distance, the satellites in all three constellations are in the LoS region. Suppose two satellites on opposite sides of the earth need to communicate. Since Kuiper's satellites will be distributed in three different altitude orbits, we approximate that all satellites are distributed in the 610 \,Km orbit. For the last five parameters, the left and right sides of the slash correspond to $\varepsilon=0.1/0.01$, respectively. 
\par
Since the latency of satellite communication is usually tens to hundreds of milliseconds, it is necessary to consider the calculation delay of the algorithm and the delay of search. The complexity of iterative method for deriving number of hops is linear. Iterations can end in a finite number of steps, and the number of hops should satisfy,
\begin{equation}\label{max iterations}
    n \leq \frac{\ln\left( 1 - \varepsilon \right)}{\ln \left( 1- \frac{1}{2} \left( 1 - \cos \frac{\theta_{\max}}{2}  \right)^{N_{\rm{Sat}}} \right)}.
\end{equation}
It can be seen that the number of iterations mostly ranges from 1 to 6. The expected contact angle and reliable angle are used to analyze the area of the search region. According to the description of the nearest neighbor search algorithm, traversing all satellites can only stay at the theoretical level. In practice, since satellite systems are massive and moving, it is difficult for a single satellite to get global information. Therefore, it is more meaningful to analyze the required area for finding a satellite than the algorithm complexity. The expectation of contact angle can be derived from the following simple derivation,
\begin{equation}
\begin{split}
\mathbb{E}\left[\theta_0\right] &= \int_0^{\theta_{\max}} 1 - F_{\theta_0}\left(\theta\right)\mathrm{d}\theta =  \int_0^{\theta_{\max}} \left(\frac{1+\cos\theta}{2}\right)^{N_{\rm{Sat}}}
\mathrm{d}\theta\\
& = 2\int_0^{\frac{\theta_{\max}}{2}}\left(\cos\widetilde{\theta} \right)^{2N_{\rm{Sat}}} \mathrm{d}\widetilde{\theta} \overset{(a)}{\approx} \pi \prod \limits_{i=1}^{N_{\rm{Sat}}} \frac{2i-1}{2i},
\end{split}
\end{equation}
where (a) follows Wallis' integrals, since $1-F_{\theta_0}(\theta)$ is very close to 0 when $\theta>\theta_{\max}$ \cite{Al-1}, the result can be approximated by continuation of the domain. Assume the spherical caps with radius of the expected contact angle and reliable angle as the average search area and maximum search area required for finding a satellite. This region is chosen as a spherical cap for computational convenience. Taking Starlink as an example, for a ten-hop link, the average search area is 0.066\% of the entire spherical area. The maximum search area is no more than 0.58\% of the spherical area. The minimum deflection angle strategy needs to search along the belt region near the shortest inferior arc. The maximum step size strategy needs to search in the whole communication range. When the reliable region is not set, the search area of the maximum step size strategy is 10.2\% of the entire spherical area for a ten-hop link. In conclusion, the proposed algorithm can end in a linear number of iterations and generally only takes a few iterations. It requires a tiny search area and has a huge advantage over other strategies. At last, note that shape of the search area is not necessary to a spherical cap, as well as surface of arbitrary shape. Since satellites are uniformly distributed on the sphere, the probability of finding a satellite is a constant for a given surface area for search.
\par
The minimum number of satellites required is obtained by testing several sets of $\theta_t$ according to proposition~\ref{Minimum Nsat} and taking the smallest of them. The probability of type-\uppercase\expandafter{\romannumeral2} interruption is obtained by Monte Carlo method: (\romannumeral1) running the algorithm for $10^6$ rounds, (\romannumeral2) recording the number of interrupt rounds and (\romannumeral3) dividing the number of interrupt rounds by $10^6$ as the probability of interruption. It can be seen that as long as the number of satellites obtained by any set of $\theta_t$ is smaller than the number of satellites in the actual constellation, type-\uppercase\expandafter{\romannumeral2} interruption does not occur. The opposite may not be accurate. For example, in the Kuiper constellation, when $\varepsilon=0.01$ and number of hops is 8, the required number of satellites obtained is 5544, which exceeds the number of satellites of the Kuiper constellation 3236. However, the second type of error still does not occur. 
\par
The last discussion is about type-\uppercase\expandafter{\romannumeral1} interruption. For Oneweb constellation with parameter $\varepsilon=0.1$, we get $n \leq 68.077$ from (\ref{max iterations}) the number of iterations reached 61. When $\varepsilon=0.01$, the iterations do not start because the reliable angle $\theta_0=0.2026$ exceeded half of the maximum dome angle $\frac{\theta_{\max}}{2} = 0.1994$. Both situations lead to the type-\uppercase\expandafter{\romannumeral1} interruption, which further leads to the occurrence of type-\uppercase\expandafter{\romannumeral2} interruption. In addition, the algorithm has high efficiency for all constellations.

\subsection{Comparison of Different Approximations}
As shown in Fig.~\ref{fig:approximation}, the performances of the two estimation methods are compared, and the relationships between latency and constellation parameters are described.
\begin{figure}[t]
	\centering
	\includegraphics[width=0.7\linewidth]{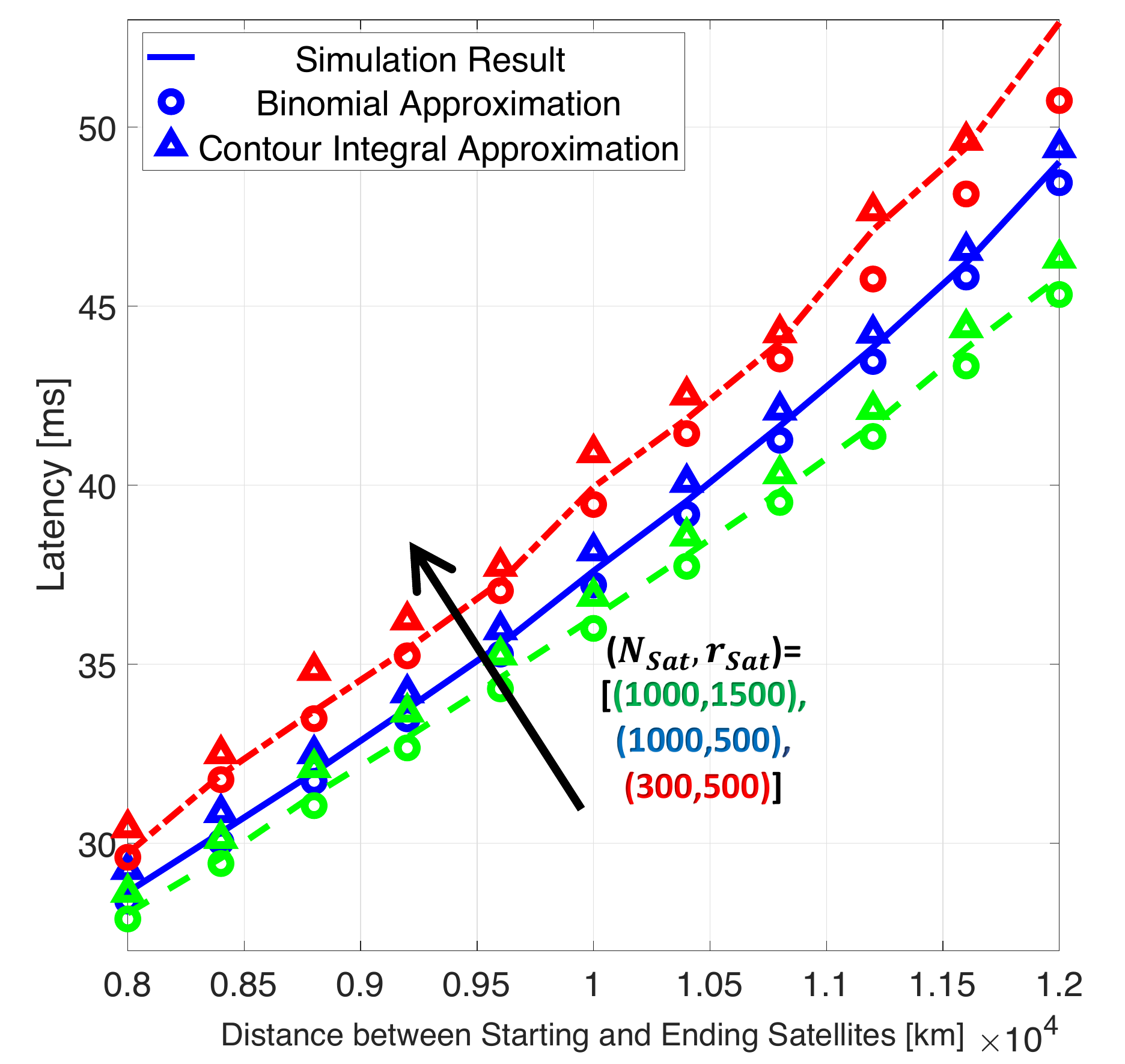}
	\caption{Comparison of different approximations.}
	\label{fig:approximation}
	\vspace{-0.4cm}
\end{figure}
In Fig.~\ref{fig:approximation}, link tolerable probability of interruption $\varepsilon=0.01$, $d_{\max}=3000\rm{\,\,Km}$, the simulation result is the exact latency obtained by Monte Carlo simulation. Both approximation methods are accurate under different constellation altitudes, satellite numbers, and link distances. Under the existing groups of parameters, binomial approximation provides a tight lower bound for the latency. At the same time, the contour integral approximation gives a tight upper bound for the latency. The accuracy of the two approximations is reduced for scenarios where the number of satellites corresponding to the red dot and dash is insufficient. Especially when the distance between the starting satellite and the ending satellite is large, binomial approximation has a relatively large gap with the actual results for the red line. As the communication distance increases and the number of satellites is insufficient, the probability of link interruption increases. In this case, the introduction of the minimum deflection strategy brings larger latency.
\par
Use the solid blue line (1000 satellites and $500\rm{ \,Km}$ constellation altitude) in Fig.~\ref{fig:approximation} as a baseline. When the communication distance is fixed, the latency is negatively correlated with the number of satellites and positively correlated with the constellation height. The decrease in the number of satellites lead to the locations of the found satellite deviating from the ideal optimal relay location, which increases latency. Although the increase of constellation height also causes the satellite location to deviate from the expected position, reducing the shortest inferior arc length has a more significant effect on the latency. A similar view can be found in proposition~\ref{prop1}. In addition, the latency increases almost linearly with the increase of communication distance, and the constellation with larger latency has a larger slope of growth.
\begin{figure}[t]
	\centering
	\includegraphics[width=0.7\linewidth]{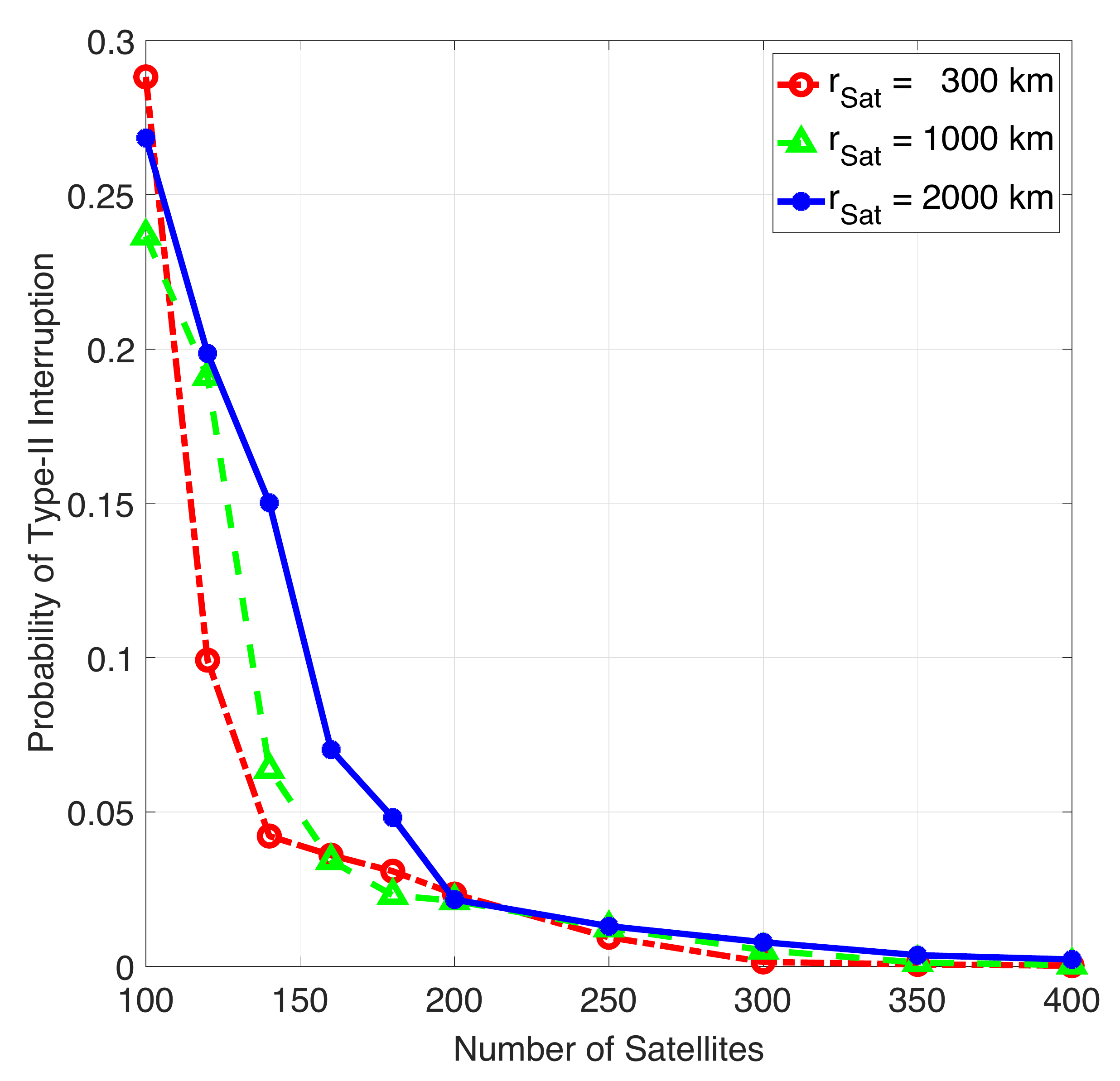}
	\caption{Probability of type-\uppercase\expandafter{\romannumeral2} interruption under different parameters.}
	\label{fig:type 2 interruption}
	\vspace{-0.4cm}
\end{figure}
\par
Fig.~\ref{fig:type 2 interruption} further explains the results in Fig.~\ref{fig:approximation} through numerical results. In Fig.~\ref{fig:type 2 interruption}, the communication distance is $10000\rm{ \,Km}$ and $d_{\max}=3000\rm{ \,Km}$. When  number of satellites $N_{\rm{Sat}}>400$, type-\uppercase\expandafter{\romannumeral2} interruption rarely occurs. When $N_{\rm{Sat}}<200$, the probability of type-\uppercase\expandafter{\romannumeral2} interruption is significantly increased with the decrease of $N_{\rm{Sat}}$ and the increase in constellation height $r_{\rm{Sat}}$. This suggests that when satellites are insufficient, the probability of type-\uppercase\expandafter{\romannumeral2} interruption is closely related to the number of satellites per unit sphere area. Furthermore, the influence of $r_{\rm{Sat}}$ on the probability of type-\uppercase\expandafter{\romannumeral2} interruption is not as significant as $N_{\rm{Sat}}$, especially when $N_{\rm{Sat}}<200$.

\subsection{Comparison of Different Strategies}
Fig.~\ref{fig:800_550} and Fig.~\ref{fig:100_550} provide the results of latency changing with distance between starting and ending satellites for different strategies. In both figures, latitude is fixed as $500\rm{ \,Km}$ and $d_{\max}=3000\rm{ \,Km}$. The number of satellites in Fig.~\ref{fig:800_550} is sufficient (800 satellites) while the number of satellites in Fig.~\ref{fig:100_550} is insufficient (100 satellites). 

\begin{figure}[t]
	\centering
	\includegraphics[width=0.7\linewidth]{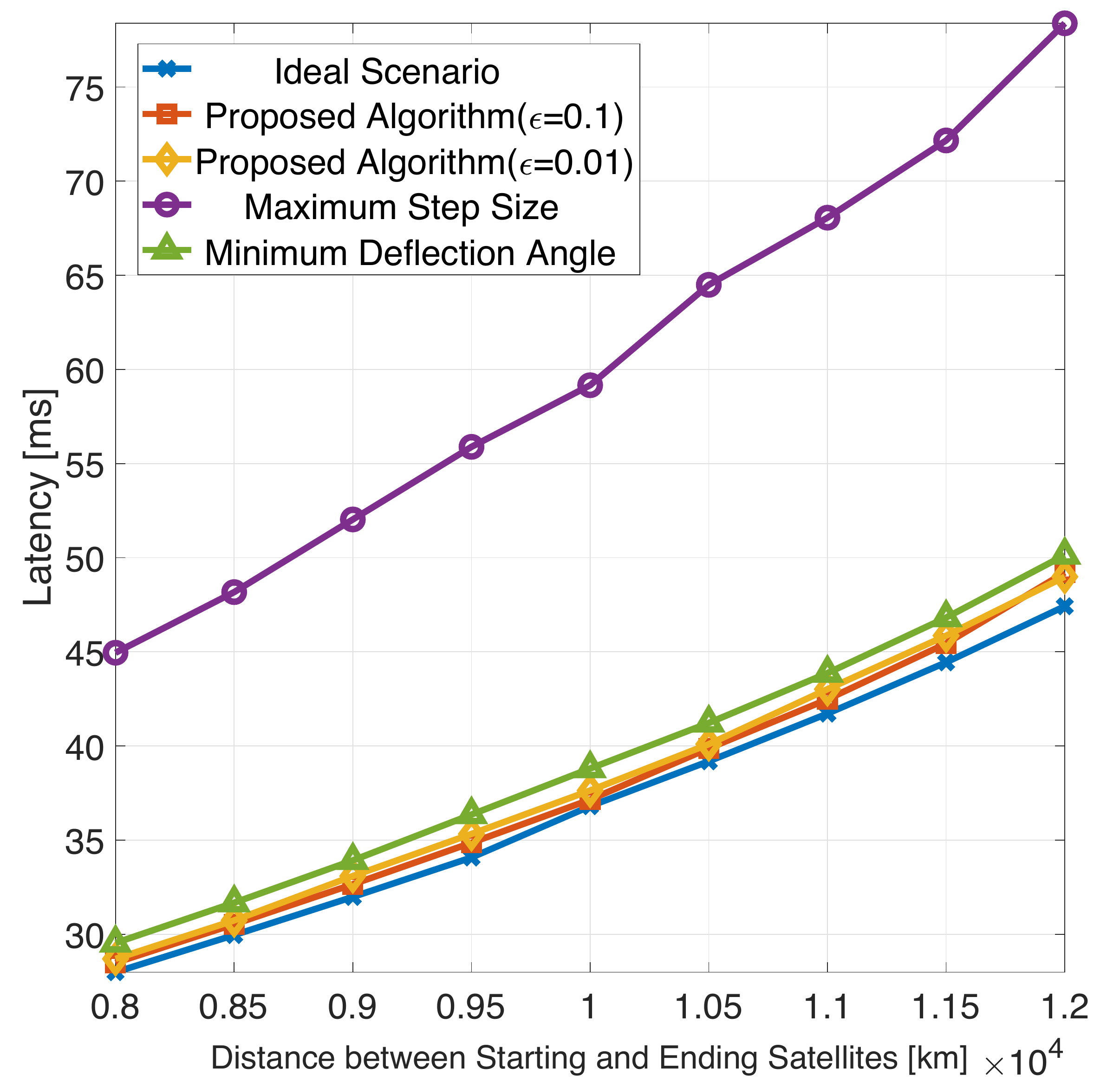}
	\caption{Influence of communication distance on different strategies ($N_{\rm{Sat}}=800$).}
	\label{fig:800_550}
	\vspace{-0.4cm}
\end{figure}

In terms of latency, the optimal scenario, the proposed algorithm, the minimum deflection angle strategy, and the maximum step size strategy are sequentially ranked from small to large. When the number of satellites is sufficient, the latency of the maximum stepsize strategy is much larger than that of other methods. The minimum deflection angle strategy and the proposed algorithm's performances are close to the lower bound. When the number of satellites is insufficient, the proposed algorithm has a remarkable advantage over the minimum deflection angle strategy. For different tolerance rates, $\varepsilon=0.01$ performs better with fewer satellites, while $\varepsilon=0.1$ performs better when satellites are sufficient.

\begin{figure}[t]
	\centering
	\includegraphics[width=0.7\linewidth]{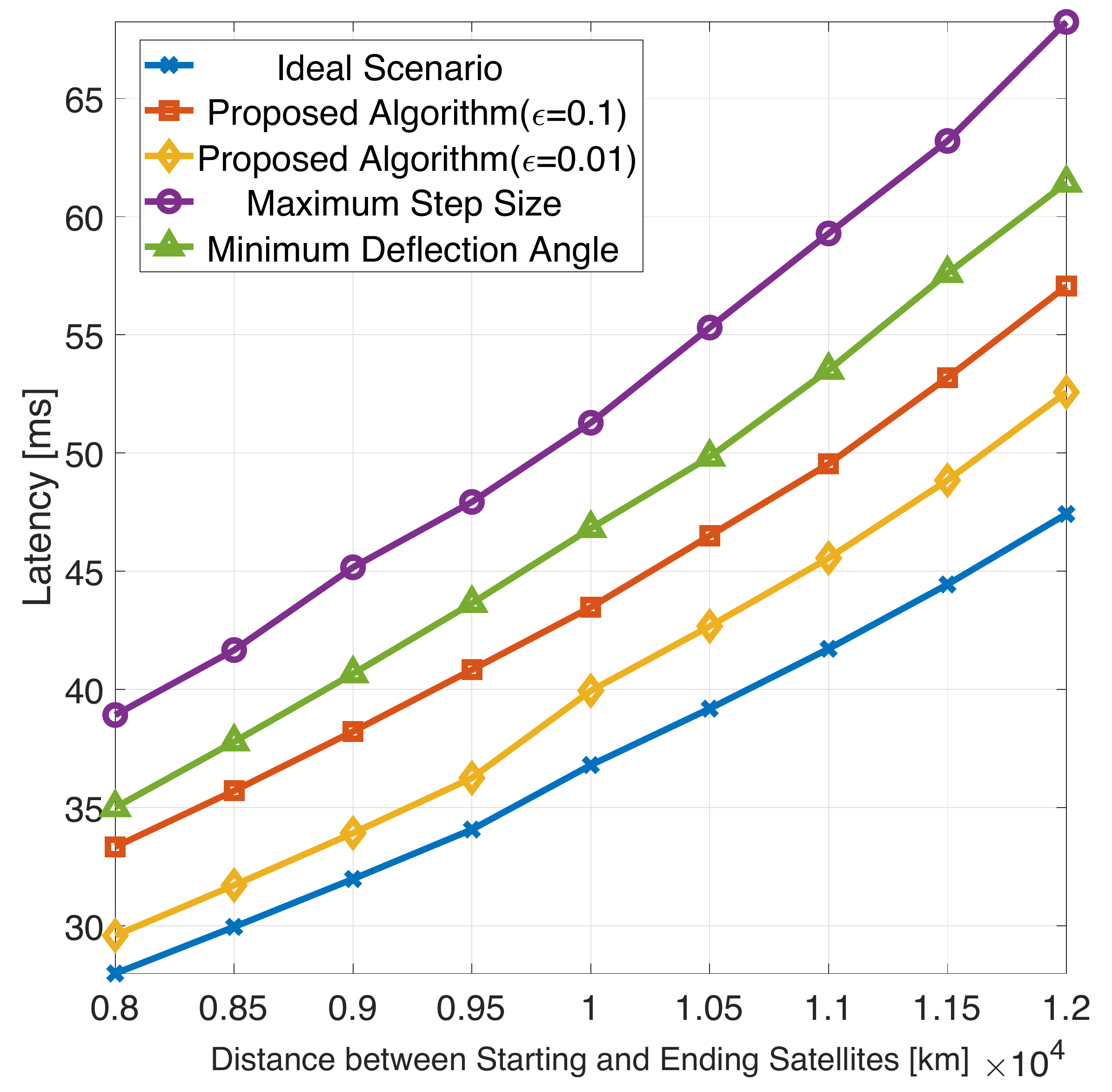}
	\caption{Influence of communication distance on different strategies ($N_{\rm{Sat}}=100$).}
	\label{fig:100_550}
	\vspace{-0.4cm}
\end{figure}

Fig.~\ref{fig:800_10000} considers the scenario where latency varies with constellation height. The number of satellites and is fixed as 800, the communication distance is fixed as $10000\rm{ \,Km}$ and $d_{\max}=3000\rm{ \,Km}$. Overall, the performance of the methods is similar to that in Fig.~\ref{fig:800_550}. The main difference is that for the proposed algorithm and the maximum step size strategy, the latency decreases with the height of the constellation. The change of the minimum deflection angle strategy is not obvious.

\begin{figure}[t]
	\centering
	\includegraphics[width=0.7\linewidth]{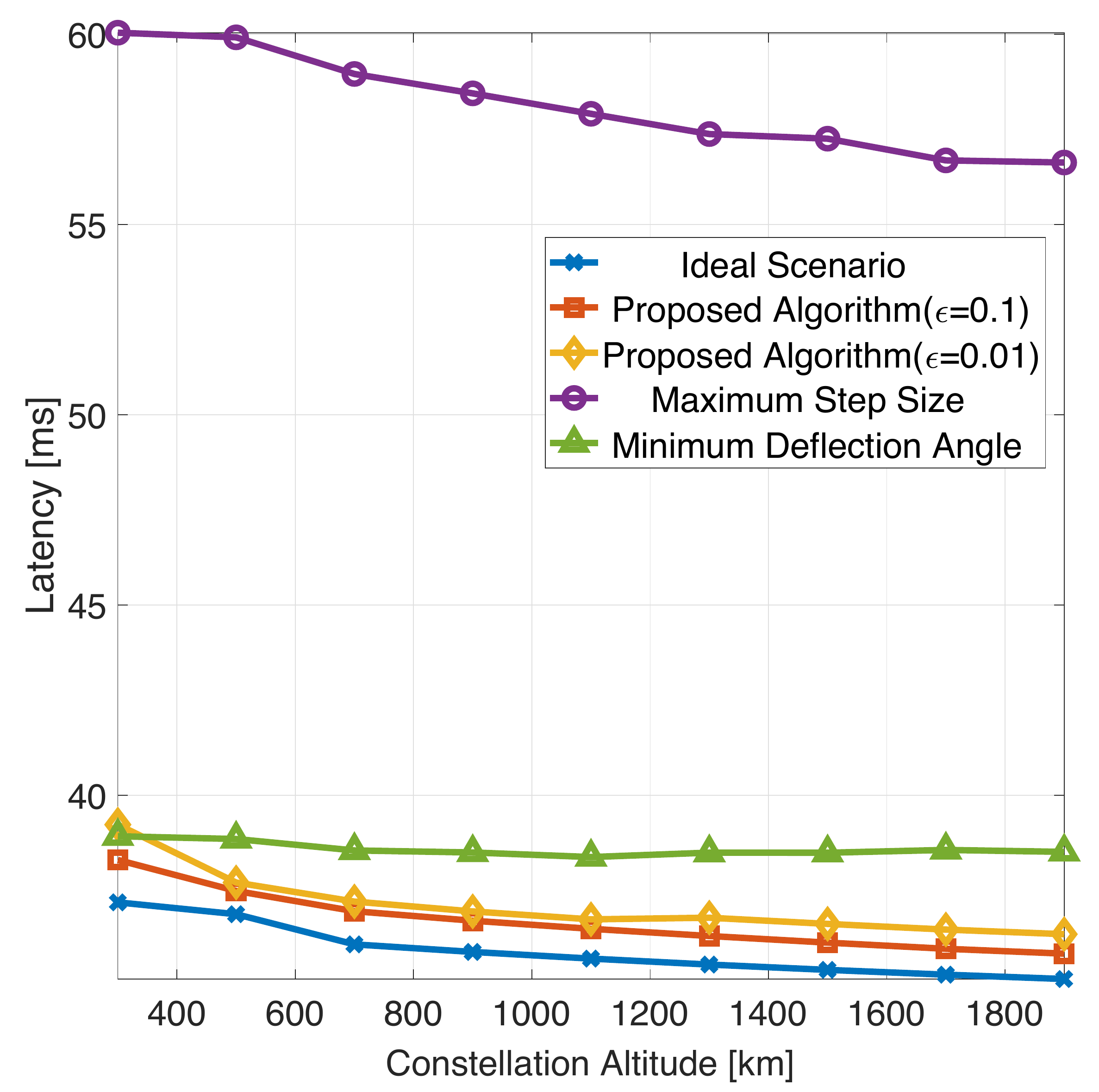}
	\caption{Influence of constellation altitude on different strategies.}
	\label{fig:800_10000}
	\vspace{-0.4cm}
\end{figure}

\section{Further Extensions}
Since the shortest routing problem on a three-dimensional sphere is not an easy problem to deal with, we simplify the model for the convenience of analysis. Although our simple model has limitations when facing some practical issues, the model is fortunately extensible.

\subsection{Expansion to multi-tier networks}
Practically, LEO satellites may assist ground base stations \cite{homssi2021modeling} with global coverage or rely on ground gateways \cite{talgat2020stochastic} to communicate. In addition, satellite systems at different altitudes (including those in synchronous orbits) also interact, such as satellites in the Kuiper constellation at three different altitudes. Therefore, cross-tier communication scenarios should be considered.
\par
Hence, it is required to investigate routing in a spherical heterogeneous network consisting of ground stations, high altitude platforms (HAP), and multi-tier LEO satellites, where satellite communications start and end with ground stations. The theoretical analysis in this paper is basically applicable to the above heterogeneous network, with the following three major changes. Firstly, the values of some parameters such as maximum communication distance $d_{\max}$ vary with different types of the relay device. This means that global information will be harder to obtain and store for ground stations.
\par
Secondly, as an essential parameter in analyzing the efficiency and reliability of the proposed algorithm, the expression and domain of the contact angle in a multi-tier network have minor modifications. Specifically, the contact angle will be replaced by the conditional contact angle, which is the contact angle of satellites distributed within the reliable communication range of both the previous and next hop.
\par
Finally, in the reliability analysis, the tier on which the relay device is located affects the probability of type-\uppercase\expandafter{\romannumeral2} interruption. Therefore, discrete Markov networks, state transition matrices, and absorption states are recommended for reliability analysis. Note that routine starts and ends on the ground stations, thus the first, middle, and last hops of the network need to be designed differently.

\subsection{Latency of Computation and Search}
Only transmission latency is considered as the objective function in this article. Computation and search latency should also be taken into account. As is mentioned, although the proposed algorithm provides a low computational complexity solution for finding the shortest latency routing on the closed sphere, its computational complexity still reaches $\mathcal{O}(N_{\min} \cdot N_{\rm{Sat}})$. The latency corresponding to this computational complexity is still large for a real-time routing with a total transmission latency of tens of milliseconds. The algorithm complexity can be reduced to $\mathcal{O}(N_{\min})$ through any of the following two schemes since only steps (5) - (9) in algorithm 2 need to be executed for both of the schemes. 
\par
When ground stations are available, we can sacrifice storage space on the ground stations for less latency. A specific data structure called Two Line Elements (TLEs) can store the dynamic positions of the satellites, and the IDs of satellites around the target position can be quickly found by index when a routing task arrives. One possible disadvantage of this scheme is that when the source is not the ground equipment but the satellite, the source needs to spend extra latency to communicate with the ground equipment. 
\par
The second scheme applies to scenarios where ground stations are unavailable. The satellite transmits a signal to the target position (obtained in step (4) - (9) of algorithm 2), and the next-hop satellite within the beam  forward this information in the above method and respond to the previous hop. Similarly, this scheme also includes extra search latency related to the contact angle, reliable angle and beamwidth. When the satellite does not receive a response from the next hop, it assumes no satellite in the beam and continues to send messages to surrounding areas. In addition, when several satellites receive the messages from the previous hop and are busy, it requires short-distance communication to schedule a single satellite for routing.

\subsection{Outage Probability and Buffering Latency}
When considering power limits, the probability of interruption and latency are related not only to distance but also to transmission signal power. Under the assumption that regenerative hops are used, a longer single-hop distance and a lower transmission power result in a larger probability of interruption and buffering latency. Under this circumstance, the \ac{SINR} serves as a bridge between them. Since satellites are less dense than ground networks and the beam is highly directional, the interference caused by other satellites can be approximated to a small constant. Assuming that the path loss of single-hop satellite-satellite channel follows the free-space fading model, the average SINR is a decreasing function of to the single-hop distance squared. 
\par
Different from the qualitative analysis before, outage probability can be a quantitative substitute for type-\uppercase\expandafter{\romannumeral2} interruption and single-hop maximum reliable distance $d_{\max}$. 
The outage probability is defined as the probability of receiving \ac{SINR} smaller than a predefined threshold $\mathbb{P}[\rm{SINR}<\gamma]$. The maximum step size proposition may not be optimal because a long single-hop distance may lead to a high probability of communication failure \cite{haenggi2005routing}. Because of the randomness of fading, signal interruptions always occur, and the retransmission mechanism can be introduced \cite{routingimportant}. 
\par
Average achievable rate is regarded as an upper bound on the as the upper bound of the transmission rate and the lower bound of the buffering latency. It is defined as the ergodic capacity from the Shannon-Hartley theorem over a fading communication link \cite{wang2021ultra}, which is proportional to $\log_2\left(1+\rm{SINR}\right)$. When the packet size is much larger than the maximum amount of data transmitted per millisecond under the average achievable rate, buffering latency is necessary to be taken into account. In order to decrease the buffering latency, a large data packet can be divided into parts and transmitted in multiple separate paths. The number of paths is determined by traffic and the average achievable rate of the relay satellites. According to proposition~\ref{prop1}, the path corresponding to the inferior arc with a smaller central angle is selected preferentially.

\subsection{Small Satellite Swarms and Storage-and-Forward Communication}
In the case of an insufficient number of satellites swarms with large packet sizes  \cite{nag2020designing}, the accessibility of data transmission is restricted, and it is challenging to realize real-time communication. These networks are demonstrated as delay/disruption tolerant networks (DTN), in which satellites store information for an amount of time after receiving it \cite{7555254}. The proposed algorithm can be extended to reduce the latency of networks with sufficient interactions. For example, with the accessibility of Earth-to-satellite links, the proposed algorithm applies to Earth observation satellite constellations. 
\par
Furthermore, small satellite swarms can help update the satellite's information (such as positions) around the relay satellite, which is beneficial for the proposed algorithm in this paper that relies on information interaction. The strategy combining the proposed algorithm with store-and-forward communication is also extendable to small spacecraft swarms communicating for interstellar exploration \cite{PARKIN2018370}. 

\section{Conclusion}
The latency minimization of multi-hop satellite links under the maximum distance constraints is studied. We propose a nearest neighbor search algorithm to determine the number of hops of multi-hop links and the position of the relay satellite in each hop. Numerical results show that the algorithm achieves linear complexity and can complete iteration in finite steps. At the same time, the search area required by the algorithm only accounts for a tiny part of the whole sphere area. The latency performance of this algorithm is very close to the minimum latency in the ideal scenario. Take Starlink constellation for example, the algorithm only needs two iterations and searches 0.066\% of the entire spherical area. The extra latency it needs to pay is no more than 1\% of the total latency of the optimal case. Furthermore, two approximations are provided to estimate the maximum gap between the latency of the proposed algorithm and the lower bound of the latency in the ideal scenario. They provide tight upper and lower bounds for latency in most cases. Finally, the influence of system parameters on multi-hop link latency is studied.

\appendices
\section{Proof of Proposition~\ref{prop1}}\label{app:prop1}
Among all circles passing $x_{h_0}$ and $x_{h_n}$ on the sphere where the satellites are located, the circle centered at the origin has the largest radius. Therefore, the shortest inferior arc divided by these two points has the smallest central angle. Based on the fact that the smaller the central angle, the shorter the length of the arc, this inferior arc has the shortest length among all arcs passing through $x_{h_0}$ and $x_{h_n}$. 
\par
For an arbitrary routing scheme, as shown in Fig.~\ref{fig:System model}, we can always locate the corresponding relay satellite on the shortest inferior arc to achieve lower latency. The correspondence of satellite positions between the two schemes is shown in Fig.~\ref{fig:System model}. In the scheme corresponding to the sky blue arrow, the distance of each hop is no longer than that of the scheme corresponding to the green arrow. Note that all subsequent concepts related to the central angle refer to the dome angle unless otherwise stated.

\section{Proof of Proposition~\ref{prop2}}\label{app:prop2}
Use (\ref{opt2}) and (\ref{st:constraint2-3}) to construct the Lagrange function,
\begin{sequation}
    \mathcal{L}\left(\theta_1^h,\theta_2^h,...,\theta_n^h\right)=\frac{1}{c}\sum_{i=1}^{n} 2r\sin\left(\frac{\theta_i^h}{2}\right)+\lambda\left(\sum_{i=1}^{n} \theta_i^h - \theta_{_{02n}}^h\right),
\end{sequation}
take the partial derivative with respect to $\theta_i^h$, we get
\begin{equation}
    \frac{\partial \mathcal{L}\left(\theta_1^h,\theta_2^h,...,\theta_N^h\right)}{\partial \theta_i^h}=\frac{r}{c}\cos\left(\frac{\theta_i^h}{2}\right)+\lambda,
\end{equation}
set the result of the partial derivative to $0$, the optimal $\theta_i^{h*}$ is
\begin{equation}
    \theta_i^{h*} = 2\arccos\left(-\frac{\lambda c}{r}\right),
\end{equation}
which is not related to $i$. Finally, the proof can be completed by combining the constraint (\ref{st:constraint2-3}).

\section{Proof of Proposition~\ref{prop3}}\label{app:prop3}

Assume that the satellites keep equal dome angles on the shortest inferior arc. The latency can be expressed as,
\begin{equation}
    T = \frac{2r}{c}\sum_{i=1}^{n} \sin\left(\frac{\theta_{_{02n}}^h}{2n}\right)=\frac{2rn}{c}\sin\left(\frac{\theta_{_{02n}}^h}{2n}\right),
\end{equation}
take partial derivative with respect to $n$,
\begin{equation}
    \frac{\partial T}{\partial n}=\frac{2r}{c}\sin \left(\frac{\theta_{_{02n}}^h}{2n}\right)-\frac{\theta_{_{02n}}^hr}{cn}\cos \left(\frac{\theta_{_{02n}}^h}{2n}\right),
\end{equation}
since an inferior arc is chosen, $\theta_{_{02n}}^h/\left(2n\right) < \pi/2$, when $n \neq 1$, we have $\cos\left({\theta_{_{02n}}^h}/{2n}\right)>0$, and
\begin{equation}
    \frac{c}{2nr} \frac{\partial T}{\partial n} = \frac{1}{\cos\left({\theta_{_{02n}}^h}/{2n}\right)}\left(\tan\left(\frac{\theta_{_{02n}}^h}{2n}\right) - \frac{\theta_{_{02n}}^h}{2n}\right),
\end{equation}
for the right-hand side of the equation, $\tan\left(\frac{\theta_{_{02n}}^h}{2n}\right)>\frac{\theta_{_{02n}}^h}{2n}$ when $\theta_{_{02n}}^h/\left(2N\right) < \pi/2$. The above analysis shows that $\frac{\partial T}{\partial n}>0$ is always satisfied. As $n$ increases, the latency $T$ increases, so we need to select the minimum number of hops that satisfies the constraints (\ref{st:constraint2-1}) and (\ref{st:constraint2-2}), the upper bound of $\theta_i^{h}$ is limited as $\theta_{\max}$ defined in (\ref{theta_max}), by solving
\begin{equation}
\begin{split}
    \theta_{_{02n}}^h = \sum_{i=1}^{N_{\min}} \theta_i^h \leq N_{\min} \theta_{\max},
\end{split}
\end{equation}
and based on the fact that $N_{\min}$ is an integer, the final result is obtained.

\section{Proof of Lemma~\ref{CDF of contact}}\label{app:CDF of contact}
Start deriving the CDF of the contact angle distribution from the definition,
\begin{equation}
\label{CDF of tagged - 1}
\begin{split}
    F_{\theta_0}\left(\theta\right) & = 1 - {\mathbb{P}}\left[ \theta_0 > \theta \right]  = 1 - \mathbb{P}\left[ {\mathcal{N}\left( \mathcal{A} \right) = 0} \right] \overset{(a)}{=} 1 - \left( 1 - \frac{\mathcal{S}\left({\mathcal{A}}\right)}{4 \pi r^2} \right)^{N_{\rm{Sat}}} \\
    &\overset{(b)}{=} 1 - \left( 1 - \frac{2 \pi r \left( r-r\cos\theta \right)}{4 \pi r^2} \right)^{N_{\rm{Sat}}}  = 1 - \left( \frac{1+\cos \theta}{2} \right) ^ {N_{\rm{Sat}}},
\end{split}
\end{equation}
where $\mathcal{N}\left( \mathcal{A} \right)$ counts the number of the satellites in the spherical cap $\mathcal{A}$ shown in Fig.~\ref{fig:angles}, $\mathcal{S}\left( \mathcal{A} \right)$ is the area measure of spherical cap $\mathcal{A}$. According to step (a), for a homogeneous point process, the probability of having satellites on the spherical cap is equal to the ratio of the area of the spherical cap to the total surface area of the sphere with radius $r$. Step (b) comes from the area formula of a spherical cap, where $r-r\cos\theta$ is the height of the spherical cap. In addition, the domain of $\theta_0$ should meet the constraints. It's easy to verify that for a constellation of hundreds of satellites, $F_{\theta_0}\left(\theta_{\max}\right)$ is very close to 1 \cite{Al-1}.

\section{Proof of Lemma~\ref{reliable angle with n}}\label{app:reliable angle with n}
Since satellites' locations are assumed to be independent, the average probability interruption of each hop should be equal. For an $n$-hop link with tolerable probability of interruption $\varepsilon$, the tolerable probability of interruption of each hop is,
\begin{equation}
    \varepsilon_1 = 1 - \left( 1 - \varepsilon \right)^{\frac{1}{n}}.
\end{equation}
In the spherical cap determined by reliable angle, the probability of having a satellite should be greater than $1-\varepsilon_1$. Since the reliable angle is the minimum angle that satisfies the above constraint, it can be obtained by the definition of the contact angle CDF,
\begin{equation}\label{substitute N_h}
    1 - \left( \frac{1+\cos \theta_r\left(n\right)}{2} \right) ^ {N_{\rm{Sat}}} = \left( 1 - \varepsilon \right)^{\frac{1}{n}},
\end{equation}
transpose and take the square root of $N_{\rm{Sat}}$ times on both sides,
\begin{equation}
    \frac{1+\cos \theta_r\left(n\right)}{2} = \left( 1 - \left( 1 - \varepsilon \right)^{\frac{1}{n}} \right)^{\frac{1}{{N_{\rm{Sat}}}}},
\end{equation}
final conclusion can be reached through simple mathematical operations.

\section{Proof of Proposition~\ref{Minimum Nsat}}\label{app:Minimum Nsat}
Since the reliable angle $\theta_r\left(n\right)$ is related to the number of hops, a $\theta_t$ unrelated to n is taken as the search radius to simplify the relationship. In this case, the minimum number of hops $N_h$ is given as,
\begin{equation}\label{N_h}
    N_t = \bigg\lceil \frac{\theta_{_{02n}}^h}{\theta_{\max}-2\theta_t} \bigg\rceil + 1.
\end{equation}
$2\theta_t < \theta_{\max}$ ensures that $N_t$ is positive. Substitute (\ref{N_h}) into (\ref{substitute N_h}),
\begin{equation}
    \left( \frac{1+\cos \theta_t}{2} \right) ^ {N_{\rm{Sat}}} \leq 1 - \left( 1 - \varepsilon \right)^{1 \big/ \left( \Big\lceil \frac{\theta_{_{02n}}^h}{\theta_{\max}-2\theta_t} \Big\rceil + 1\right) },
\end{equation}
take the logarithm of both sides, and divide by the $\ln \left( \frac{1+\cos \theta_t}{2} \right)$ of both sides to get the result. Note that (\ref{substitute N_h}) guarantees that the $\theta_t$ satisfying (\ref{sufficient condition}) must be greater than or equal to the reliable angle.
A set of practical $\theta_t$ can be taken as,
\begin{equation}
    \bigg\{\frac{1}{2} \left( \theta_{\max} -  \frac{\theta_{_{02n}}^h}{N_{\min}+k} \right),k=0,1,2...\bigg\}.
\end{equation}

\section{Proof of Proposition~\ref{contour integral approximation}}\label{app:contour integral approximation}
Assuming that the contact angles between the two relay positions and their nearest satellites are $\theta_0^{(1)}$ and $\theta_0^{(2)}$, respectively. In this case, these two satellites are uniformly distributed on circles $\mathcal{O}_1(\theta_0^{(1)})$ and $\mathcal{O}_2(\theta_0^{(2)})$ with radius $r\sin\theta_0^{(1)}$ and $r\sin\theta_0^{(2)}$  respectively. The average distance at contact angles $\theta_0^{(1)}$ and $\theta_0^{(2)}$ can be obtained by contour integral around two circles with respect to single-hop distance $d_1$. Therefore, the expectation of single-hop distance $d_1$ can be expressed as,
\begin{equation}
    \mathbb{E}\left[d_1\right] = \mathbb{E}_{\theta_0^{(1)},\theta_0^{(2)}}\left[ \oint_{\mathcal{O}_1(\theta_0^{(1)})} \oint_{\mathcal{O}_2(\theta_0^{(2)})} f_{d_1} \mathrm{d}\mathcal{O}_2\mathrm{d}\mathcal{O}_1\right],
\end{equation}
where $f_{d_1}$ is the PDF of $d_1$, it is related to the contact angles $\theta_0^{(1)}$, $\theta_0^{(2)}$ and the positions on the corresponding $\mathcal{O}_1$, $\mathcal{O}_2$. The expression of $f_{d_1}$ is hard to express in either rectangular or spherical coordinates. Let us split the problem in two. One of the satellites is fixed to the relay position, while the other is uniformly distributed on the circle. The uniform distribution of a satellite can be offset by changing the position of a relay position. This amount of change can be described by $\alpha$. By symmetrically making the same change of the other relay position, the amount of change becomes $2\alpha-1$. 
\par
Since rotation does not affect the distribution of the satellite, let the spherical coordinate of the relay position be $(r,0,0)$. The coordinate of the fixed satellite is $(r,\theta^h,0)$, where $\theta^h$ is the dome angle of the single hop. Assume the contact angles between the relay position and its nearest satellite is $\theta_0$, by equation
\begin{equation}
    \frac{1}{\pi}\int_0^{\theta_{\max}} d(\theta_0,\varphi,\theta_h,0)  \mathrm{d}\varphi = \alpha\left(\theta_0,\theta^h\right) 2r\sin\left(\frac{\theta_h}{2}\right),
\end{equation}
where $d(\theta_0,\varphi,\theta_h,0)$ is defined in (\ref{d function}), the amount of change $\alpha\left(\theta_0,\theta^h\right)$ can be written as,
\begin{equation}
\begin{split}
    \alpha\left(\theta_0,\theta^h\right) = \frac{\sqrt{2}}{2\pi\sin\frac{\theta^h}{2}}  \int_0^{\pi} \sqrt{1-\cos\theta_0\cos\theta_h-\sin\theta_0\sin\theta^h\cos\varphi}\, \mathrm{d}\varphi.
\end{split}
\end{equation}
Note that for a small $\theta_0$,
\begin{equation}
\begin{split}
   \alpha\left(\theta_0,\theta^h\right) \approx \frac{\sqrt{2(1-\cos\theta_h)}}{2\sin{\frac{\theta^h}{2}}} = 1.
   \end{split}
\end{equation}
Take the expectation of $\alpha\left(\theta_0,\theta^h\right)$ with respect to $\theta_0$,
\begin{equation}
    \overline{\alpha}\left(\theta^h\right) = \int_0^{\theta_{\max}}f_{\theta_0}(\theta)\alpha\left(\theta_0,\theta^h\right)\mathrm{d}\theta,
\end{equation}
the result in (\ref{overline alpha}) is derived. Since the propagation speed of the laser is constant, the ratio of latency is equivalent to the ratio of distance, the proof of theorem~\ref{contour integral approximation} is finished.

\bibliographystyle{IEEEtran}
\bibliography{references}

\end{document}